\newif\ifdraft \drafttrue
\newif\iffull \fulltrue

\iffull
\documentclass[11pt]{article}
\else
\documentclass{sig-alternate}
\fi

\newenvironment{mathdisplayfull}{\iffull \[ \else $ \fi}{\iffull \]\else $ \ignorespaces\fi}
\newcommand{\shortbreak}{{\iffull \else \\ \fi}}
\newcommand{\shortpagebreak}{{\iffull \else \vfill\eject \fi }}
\newcommand{\longproof}{{\iffull Proof \fi}}
\newcommand{\longquad}{{\iffull \qquad \else \;\;  \fi}}
\newcommand{\fullfrac}[2]{\iffull \frac{#1}{#2} \else #1 / #2 \fi}

\iffull
\newcommand{\SUBSECTION}{\subsection}
\else
\newcommand{\SUBSECTION}{\paragraph*}
\fi

\usepackage{algorithm, algorithmic}
\iffull
\usepackage{amsmath, amssymb, amsthm}
\usepackage{fullpage}
\fi
\usepackage{verbatim}
\usepackage{color}
\usepackage{times}
\definecolor{DarkGreen}{rgb}{0.1,0.5,0.1}
\definecolor{DarkRed}{rgb}{0.5,0.1,0.1}
\definecolor{DarkBlue}{rgb}{0.1,0.1,0.5}
\usepackage[]{hyperref}
\hypersetup{
    unicode=false,          
    pdftoolbar=true,        
    pdfmenubar=true,        
    pdffitwindow=false,      
    pdftitle={},    
    pdfauthor={}
    pdfsubject={},   
    pdfnewwindow=true,      
    pdfkeywords={keywords}, 
    colorlinks=true,       
    linkcolor=DarkRed,          
    citecolor=DarkGreen,        
    filecolor=DarkRed,      
    urlcolor=DarkBlue,          
}

\ifdraft
\newcommand{\mynote}[1]{\marginpar{\tiny\sf #1}}
\else
\newcommand{\mynote}[1]{}
\fi

\newcommand\N{\mathbb{N}}
\newcommand\R{\mathbb{R}}
\newcommand\cA{\mathcal{A}}
\newcommand\cB{\mathcal{B}}
\newcommand\cE{\mathcal{E}}
\newcommand\cF{\mathcal{F}}

\newcommand\cQ{\mathcal{Q}}
\newcommand\cR{\mathcal{R}}
\newcommand\cX{\mathcal{X}}

\newcommand\cU{\mathcal{U}}

\newcommand{\cI}{\mathcal{I}}

\newcommand{\bfQ}{\mathbf{Q}}

\newcommand\bits{\{0,1\}}
\newcommand{\getsr}{\gets_{\mbox{\tiny R}}}
\newcommand\set[1]{\left\{#1\right\}} 
\newcommand{\from}{:}
\newcommand{\dist}[1]{\widetilde{#1}}
\newcommand{\proj}[1]{\Gamma #1}
\newcommand{\projs}[1]{\Gamma_{s} #1}
\renewcommand{\tilde}{\widetilde}

\newcommand{\ex}[1]{\mathbb{E}\left[#1\right]}
\DeclareMathOperator*{\Expectation}{\mathbb{E}}
\newcommand{\Ex}[2]{\Expectation_{#1}\left[#2\right]}

\newcommand{\prob}[1]{\mathrm{Pr}\left[#1\right]}

\newcommand{\san}{\cA}

\newcommand{\univ}{\cX}

\newcommand{\eps}{\varepsilon}

\newcommand{\Lap}{\mathrm{Lap}}
\renewcommand{\hat}{\widehat}

\newcommand{\id}{\mathrm{id}}

\newcommand{\INDSTATE}[1][1]{\STATE\hspace{#1\algorithmicindent}}

\newtheorem{theorem}{Theorem}[section]
\newtheorem{lemma}[theorem]{Lemma}

\newtheorem{claim}[theorem]{Claim}
\newtheorem{remark}[theorem]{Remark}
\newtheorem{corollary}[theorem]{Corollary}

\iffull
\theoremstyle{definition}
\fi
\newtheorem{definition}[theorem]{Definition}

\title{Differential Privacy for the Analyst via \\ Private Equilibrium
  Computation\thanks{A full version of this work appears on
    arXiv~\cite{HRU12}.}}
\iffull
\author{
  Justin Hsu\thanks{
    University of Pennsylvania Department of Computer and Information Sciences.
    Supported in part by NSF grant CNS-1065060.
    Email: \href{mailto:justhsu@cis.upenn.edu}{justhsu@cis.upenn.edu}.}
  \and Aaron Roth\thanks{
    University of Pennsylvania Department of Computer and Information Sciences.
    Supported in party by NSF grant CCF-1101389 and CNS-1065060.
    Email: \href{mailto:aaroth@cis.upenn.edu}{aaroth@cis.upenn.edu}.}
  \and Jonathan Ullman\thanks{
    Harvard University School of Engineering and Applied Sciences.
    Supported by NSF grant CNS-1237235 and a Siebel Scholarship.
    Email: \href{mailto:jullman@seas.harvard.edu}{jullman@seas.harvard.edu}.}}
\else
\numberofauthors{3}
\author{
\alignauthor
Justin Hsu\thanks{Supported in part by NSF grant CNS-1065060.} \\
\affaddr{University of Pennsylvania} \\
\email{justhsu@cis.upenn.edu}
\and
\alignauthor
Aaron Roth\thanks{Supported in part by NSF grant CCF-1101389 and CNS-1065060.} \\
\affaddr{University of Pennsylvania} \\
\email{aaroth@cis.upenn.edu}\\
\alignauthor
Jonathan Ullman\thanks{Harvard University School of Engineering and Applied
Sciences. Supported by NSF grant CNS-1237235 and a Siebel Scholarship.} \\
\affaddr{Harvard University} \\
\email{jullman@seas.harvard.edu}
}
\fi

\begin{document}

\maketitle

\begin{abstract}
We give new mechanisms for answering exponentially many queries from multiple
analysts on a private database, while protecting differential privacy both for
the individuals in the database and for the analysts. That is, our mechanism's
answer to each query is nearly insensitive to changes in the queries asked by
other analysts.  Our mechanism is the first to offer differential privacy on the
joint distribution over analysts' answers, providing privacy for data analysts
even if the other data analysts collude or register multiple accounts. In some
settings, we are able to achieve nearly optimal error rates (even compared to
mechanisms which do not offer analyst privacy), and we are able to extend our
techniques to handle non-linear queries. Our analysis is based on a novel view
of the private query-release problem as a two-player zero-sum game, which may be
of independent interest.
\end{abstract}

\section{Introduction}

Consider a tracking network that wants to sell a database of consumer data to
several competing analysts conducting market research.  The administrator of the
tracking network faces many opposing constraints when deciding how to provide
analysts with this data. For legal reasons, the privacy of the individuals
contained in her database must be protected. At the same time, the analysts must
be able to query the database and receive useful answers. Finally, the privacy
of the \emph{queries} made to the database must be protected, since the analysts
are in competition and their queries may be disclosive of proprietary
strategies.

This setting of {\em analyst privacy} was recently introduced in a beautiful
paper of Dwork, Naor, and Vadhan \cite{DNV12}. They showed that differentially
private \emph{stateless} mechanisms --- which answer each query
\emph{independently} of the previous queries --- can only give accurate answers
when the number of queries is at most quadratic in the size of the database.
This result rules out mechanisms that perfectly protect the privacy of the
queries, while accurately answering exponentially many queries --- answers must
depend on the state, and hence on the previous queries. However, it turns out
that mechanisms that offer a differential-privacy-like guarantee with respect to
the queries are possible: Dwork, et al. \cite{DNV12} give such a mechanism, with
the guarantee that the marginal distribution on answers given to each analyst is
differentially private with respect to the set of queries made by all of the
other analysts.  Their mechanism is capable of answering exponentially many
linear queries with error $\tilde{O}(1/n^{1/4})$, where $n$ is the number of
records in the database.  A {\em linear query} is a $(1/n)$-sensitive query of
the form ``What fraction of the individual records in the database satisfy some
property $q$?'', so their mechanism gives non-trivial accuracy.

However, they note that their mechanism has several shortcomings. First, it does
not promise differential privacy on the \emph{joint distribution} over multiple
analysts' answers. Therefore, if multiple analysts collude, or if a single
malicious analyst registers several false accounts with the mechanism, then the
mechanism no longer guarantees query privacy. Second, their mechanism is less
accurate than known, non-analyst private mechanisms --- analyst privacy is
achieved at a cost to accuracy. Finally, the mechanism can only answer linear
queries, rather than general low-sensitivity queries.

In this paper, we address all of these issues. First, we consider mechanisms
which guarantee \emph{one-query-to-many-analyst} privacy: for each analyst $a$,
the joint distribution over answers given to \emph{all other} analysts $a' \neq
a$ is differentially private with respect to the change of a single query asked
by analyst $a$.  This privacy guarantee is incomparable to that of Dwork, et al.
\cite{DNV12}: it is weaker, because we protect the privacy of a single query,
rather than protecting the privacy of all queries asked by analysts $a' \neq a$.
However, it is also stronger, because the privacy of one query from an analyst
$a$ is preserved even if all other analysts $a' \neq a$ collude or register
multiple accounts.  Our first result is a mechanism in this setting, with error
at most $\tilde{O}(1/\sqrt{n})$ for answering exponentially many linear queries.
This error is optimal up to polylogarithmic factors, even when comparing to
mechanisms that only guarantee data privacy.

We then extend our techniques to \emph{one-analyst-to-many-analyst} privacy,
where we require that the mechanism preserve the privacy of analyst when he
changes {\em all} of his queries, even if {\em all other} analysts collude. Our
second result is a mechanism in this setting, with error $\tilde{O}(1/n^{1/3})$.
Although this error rate is worse than what we achieve for
one-query-to-many-analyst privacy (and not necessarily optimal), our mechanism
is still capable of answering exponentially many queries with non-trivial
accuracy guarantees, while satisfying both data and analyst privacy.

These first two mechanisms operate in the \emph{non-interactive} setting, where
the queries from every analyst are given to the mechanism in a single batch.
Our final result is a mechanism in the \emph{online} setting that satisfies
one-query-to-many analyst privacy.  The mechanism accurately answers a (possibly
exponentially long) \emph{fixed} sequence of low-sensitivity queries.  Although
our mechanism operates as queries arrive online, it cannot tolerate
adversarially chosen queries (i.e.~it operates in the same regime as the
\emph{smooth multiplicative weights} algorithm of Hardt and
Rothblum~\cite{HR10}).  For linear queries, our mechanism gives answers with
error at most $\tilde{O}(1/n^{2/5})$. For answering general queries with
sensitivity $1/n$ (the sensitivity of a linear query), the mechanism guarantees
error at most $\tilde{O}(1/n^{1/10})$.

When answering $k$ queries on a database $D \in \univ^n$ consisting of $n$
records from a data universe $\univ$, our offline algorithms run in time
$\tilde{O}(n \cdot (|\univ| + k))$ and our online algorithm runs in time
$\tilde{O}(|\univ| + n)$ per query.  These running times are essentially optimal
for mechanisms that answer more than $\omega(n^2)$ arbitrary linear
queries~\cite{Ullman}, assuming (exponentially hard) one-way functions exist.

\SUBSECTION{Our Techniques}
To prove our results, we take a novel view of private query release as a
two player zero-sum game between a \emph{data player} and a \emph{query
player}. For each element of the data universe $x \in X$, the data player has
an action $a_x$. Intuitively, the data player's mixed strategy will be his
approximation of the true database's distribution.

On the other side, for each query $q \in \cQ$, the query player has two actions:
$a_q$ and $a_{\neg q}$. The two actions for each query allow the query player to
penalize the data player's play, both when the approximate answer to $q$ is too
high, and when it is too low --- the query player tries to play queries for
which the data player's approximation performs poorly. Formally, we define
the cost matrix by
\begin{mathdisplayfull}
  G(a_q,a_x) = q(x) - q(D)
\end{mathdisplayfull}%
and
\begin{mathdisplayfull}
  G(a_{\neg q},a_x) = q(D) - q(x),
\end{mathdisplayfull}%
where $D$ is the private database. The query player wishes to maximize the cost,
whereas the database player wishes to minimize the cost. We show that the value
of this game is $0$, and that any $\rho$-approximate equilibrium strategy for
the database player corresponds to a database that answers every query $q \in
\cQ$ correctly up to additive error $\rho$.  Thus, given any pair of
$\rho$-approximate equilibrium strategies, the strategy for the data player will
constitute a database (a distribution over $\mathcal{X}$) that answers every
query to within error $O(\rho)$.

Different privacy constraints for the private query release problem can be
mapped into privacy constraints for solving two-player zero-sum games. Standard
private linear query release corresponds to privately computing an approximate
equilibrium, where privacy is preserved with respect to changing every cost in
the game matrix by at most $1/n$. Likewise, query release while protecting
\emph{one-query-to-many-analyst} privacy corresponds to computing an approximate
equilibrium strategy, where privacy is with respect to an arbitrary change in
two \emph{rows} of the game matrix --- changing a single query $q$ changes the
payoffs for actions $a_q$ and $a_{\neg q}$. Our main result can be viewed as an
algorithm for privately computing the equilibrium of a zero-sum game while
protecting the privacy of strategies of the players, which may be of independent
interest.

To construct an approximate equilibrium, we use a well-known result: when two
no-regret algorithms are played against each other in a zero-sum game, their
empirical play distributions quickly converge to an approximate equilibrium.
Thus, to compute an equilibrium of $G$, we have the query player and the data
player play against each other using no-regret algorithms, and output the
empirical play distribution of the data player as the hypothesis database. We
face several obstacles along the way.

First, no-regret algorithms maintain a \emph{state} --- a distribution over
actions, which is not privacy preserving.  (In fact, it is computed
deterministically from inputs that may depend on the data or queries.)  Previous
approaches to private query release have addressed this problem by adding noise
to the \emph{inputs} of the no-regret algorithm.

In our approach, we crucially rely on the fact that {\em sampling} actions from
the distributions maintained by the multiplicative weights algorithm is privacy
preserving. Intuitively, privacy will come from the fact that the multiplicative
weights algorithm does not adjust the weight on any action too aggressively,
meaning that when we view the weights as defining a distribution over actions,
changing the losses experienced by the algorithm in various ways will have a
limited effect on the distribution over actions.  We note that this property is
not used in the private multiplicative weights mechanism of Hardt and Rothblum
\cite{HR10}, who use the distribution itself as a hypothesis.  Indeed, without
the constraint of query privacy, \emph{any} no-regret algorithm can be used in
place of multiplicative weights \cite{RR10,GRU12}, which is not the case in our
setting.

Second, sampling from the multiplicative weights algorithm is private only if
the changes in losses are small. Intuitively, we must ensure that changing one
query from one analyst does not affect the losses experienced by the data player
too dramatically, so that samples from the multiplicative weights algorithm will
indeed ensure query privacy.  To enforce this requirement, we force the query
player to play mixed strategies from the set of \emph{smooth} distributions,
which do not place too much weight on any single action.  It is known that
playing any no-regret algorithm, but projecting into the set of smooth
distributions in the appropriate way (via a Bregman projection), will ensure
no-regret with respect to any smooth distribution on actions. For comparison,
no-regret is typically defined with respect to the best single action, which is
a not a smooth distribution. Thus, our regret guarantee is weaker.

The result of this simulation is an approximate equilibrium strategy for the
data player, in the sense that it achieves approximately the value of the game
when played against \emph{all but $s$} strategies of the query player, where
$1/s$ is the maximum probability that the query player may assign to any action.
This corresponds to a synthetic database, which we release to all analysts, that
answers {\em all but $s$} queries accurately. Then, since we choose $s$ to be small,
we can answer the mishandled queries with the sparse vector technique
\cite{DNRRV09,RR10,HR10} adding noise only $\tilde{O}(\sqrt{s}/n)$ to these $s$
queries. The result is a nearly optimal error rate of $\tilde{O}(1/\sqrt{n})$

Our techniques naturally extend to one-analyst-to-many-analyst privacy by making
the actions of the query player correspond to entire workloads of queries, one
for each analyst, where the query player picks analysts that have at least one
query that has high error on the current hypothesis. Like before, a small number
of analysts will have queries that have high error, which we handle with a
separate private query release mechanism for each analyst.

Finally, we use these techniques to convert the private multiplicative weights
algorithm of Hardt and Rothblum~\cite{HR10} into an online algorithm that
preserves one-query-to-many-analyst privacy, and also answers arbitrary
low-sensitivity queries. These last two extensions both give first-of-their-kind
results, but at some degradation in the accuracy parameters: we do not obtain
$O(1/\sqrt{n})$ error rate.  We leave it as an open problem to achieve
$\tilde{O}(1/\sqrt{n})$ error in these settings, or show that the accuracy cost
is necessary.

\SUBSECTION{Related Work}
There is an extremely large body of work on differential privacy \cite{DMNS06}
that we do not attempt to survey.  The study of differential privacy was
initiated by a line of work \cite{DN03,BDMN05,DMNS06} culminating in the
definition by Dwork, Mcsherry, Nissim, and Smith \cite{DMNS06}, who also
introduced the basic technique of answering low-sensitivity queries using the
Laplace mechanism. The Laplace mechanism gives approximate answers to nearly
$O(n^2)$ queries while preserving differential privacy.

A recent line of work \cite{BLR08,DNRRV09,DRV10,RR10,HR10,GHRU11,GRU12, HLM12}
has shown how to accurately answer almost \emph{exponentially many} queries
usefully while preserving differential privacy of the data. Some of this work
\cite{RR10,HR10,GRU12} rely on \emph{no-regret} algorithms --- in particular,
Hardt and Rothblum \cite{HR10} introduced the multiplicative weights technique
to the differential privacy literature, which we use centrally.

However, we make use of multiplicative weights in a different way from prior
work in private query release --- we simulate play of a two-player zero-sum game
using {\em two} copies of the multiplicative weights algorithm, and rely on the
fast convergence of such play to approximate Nash equilibrium \cite{FS}. We also
rely on the fact that Bregman projections onto a convex set $K$ can be used in
conjunction with the multiplicative weights update rule\iffull\footnote{Indeed,
  they can be used in conjunction with any no-regret algorithm in the family of
regularized empirical risk minimizers.} \fi \space to achieve no regret with
respect to the best element in the set $K$ \cite{RakhlinNotes}.  Finally, we use
that \emph{samples} from the multiplicative weights distribution can be viewed
as samples from the exponential mechanism of McSherry and Talwar \cite{MT07},
and hence are privacy preserving.

Our use of Bregman projections into smooth distributions is similar to its use
in \emph{smooth boosting}. Barak, Hardt, and Kale \cite{BHK09} use Bregman
projections in a similar way, and the weight capping used by Dwork, Rothblum,
and Vadhan \cite{DRV10} in their analysis of \emph{boosting for people} can be
viewed as a Bregman projection.

The most closely related paper to ours is the beautiful recent work of Dwork,
Naor, and Vadhan \cite{DNV12}, who introduce the idea of analyst privacy. They
show that any algorithm which can answer $\omega(n^2)$ queries to non-trivial
accuracy must maintain \emph{common state} as it interacts with many data
analysts, and hence potentially violates the privacy of the analysts.
Accordingly, they give a stateful mechanism which promises many-to-one-analyst
privacy, and achieves per-query error $\tilde{O}(1/n^{1/4})$ for linear queries
--- their mechanism promises differential privacy on the marginal distribution
of answers given to any single analyst, even when all other analysts change all
of their queries. However, if multiple analysts collude, or if a single analyst
can falsely register under many ids, then the privacy guarantees degrade
quickly --- privacy is not promised on the \emph{joint distribution} on all
analysts answers. Lifting this limitation, improving the error bounds and
extending analyst privacy to non-linear queries, are all stated as open
questions.

Finally, the varying notions of analyst privacy we use can be interpreted in the
context of two-party differential privacy, introduced by McGregor, et
al.~\cite{mcgregor}. If we consider a single analyst as one party, sending
private queries to a second party consisting of the mechanism and all the other
parties indirectly, the many-to-one-analyst privacy guarantee is equivalent
to privacy of the first party's view. Here, the privacy must be protected even
if the second party changes its inputs arbitrarily, i.e., the other analysts
change their queries arbitrarily.

On other hand, if we consider all but one analyst as the first party, sending
queries to the mechanism and the remaining analyst, one-query-to-many-analyst
privacy is equivalent to the first party's view being private when the single
analyst changes a query. One-analyst-to-many-analyst privacy is similar: the
first party's view must be private when the second party changes all of its
queries.

\section{Preliminaries}

\paragraph*{Differential Privacy and Analyst Differential Privacy}
Let a \emph{database} $D \in \cX^n$ be a collection of $n$ records (rows) \break
$\{x^{(1)}, \dots, x^{(n)}\}$ from a \emph{data universe} $\cX$. Two databases $D,D'
\in \cX^n$ are \emph{adjacent} if they differ only on a single row, which we
denote by $D \sim D'$.

A mechanism $\san: \cX^n \to \cR$ takes a database as input and outputs some
data structure in $\cR$.  We are interested in mechanisms that satisfy
\emph{differential privacy}.
\begin{definition}
  \label{def:dp}
A mechanism $\san\from \cX^n \to \cR$ is \emph{$(\eps, \delta)$-differentially
private} if for every two adjacent databases $D \sim D' \in \cX^n$ and every
subset $S \subseteq \cR$,
\[
\prob{\san(D) \in S} \leq e^{\eps} \prob{\san(D') \in S} + \delta.
\]
\end{definition}

In this work we construct mechanisms that ensure differential privacy \emph{for
the analyst} as well as for the database.  To define analyst privacy, we first
define \emph{many-analyst mechanisms}. Let $\bfQ$ be the set of all allowable
queries. The mechanism takes $m$ sets of queries $\cQ_1, \dots, \cQ_m$ and
returns $m$ outputs $Z_1, \dots, Z_m$, where $Z_j$ contains answers to the
queries $\cQ_j$.  Thus, a many-analyst mechanism has the form $\san \from
\univ^n \times (\bfQ^*)^m \to \cR^m$.  Given sets of queries $\cQ_1, \dots,
\cQ_m$, let $\cQ = \bigcup_{j=1}^{m} \cQ_j$ denote the set of all queries.  In
guaranteeing privacy even in the event of collusion, it will be useful to refer
to the output given to all analysts other than some analyst $i$.  For each $\id
\in [m]$ we write $\san(D,\cQ)_{-\id}$ to denote
$(Z_1,\ldots,Z_{\id-1},Z_{\id+1},\ldots,Z_m)$, the output given to all analysts
other than $\id$.

Let $\cQ = \cQ_1, \dots, \cQ_m$ and $\cQ' = \cQ'_1, \dots, \cQ'_m$.  We say that
$\cQ$ and $\cQ'$ are \emph{analyst-adjacent} if there exists $\id^* \in [m]$
such that for every $\id \neq \id^*$, $\cQ_{\id} = \cQ'_{\id}$.  That is, $\cQ
\sim \cQ'$ are analyst adjacent if they differ only on the queries asked by one
analyst.  Intuitively, a mechanism satisfies one-analyst-to-many-analyst privacy
if changing \emph{all the queries asked by analyst $\id^*$} does not
significantly affect the output given to \emph{all analysts other than $\id^*$}.
\begin{definition}
A many-analyst mechanism $\san$ satisfies \emph{$(\eps,
\delta)$-one-analyst-to-many-analyst privacy} if for every database $D \in
\univ^n$, every two analyst-adjacent query sequences $\cQ \sim \cQ'$ that differ
only on one set of queries $\cQ_{\id}, \cQ'_{\id}$, and every $S \subseteq
\cR^{m-1}$,
\[
\prob{\san(D, \cQ)_{-\id} \in S} \leq e^{\eps} \prob{\san(D, \cQ')_{-\id} \in S} + \delta.
\]
\end{definition}

Let $\cQ = \cQ_1, \dots, \cQ_m$ and $\cQ' = \cQ'_1, \dots, \cQ'_m$.  We say that
$\cQ$ and $\cQ'$ are \emph{query-adjacent} if there exists $\id^*$ such that for
every $\id \neq \id^*$, $\cQ_{\id} = \cQ'_{\id}$ and $| \cQ_{\id^*} \triangle
\cQ'_{\id^*} | \leq 1$.  That is, $\cQ \sim \cQ'$ are query adjacent if they
differ only on one of the queries.  Intuitively, we say that a mechanism
satisfies one-query-to-many-analyst privacy if changing \emph{one query asked by
analyst $\id^*$} does not significantly affect the output given to \emph{all
analysts other than $\id^*$}.
\begin{definition}
A many-analyst mechanism $\san$ satisfies \emph{$(\eps,
\delta)$-one-query-to-many-analyst privacy} if for every database $D \in
\univ^n$, every two query-adjacent query sequences $\cQ \sim \cQ'$ that differ
only on one query in $\cQ_{\id}, \cQ'_{\id}$, and every $S \subseteq \cR^{m-1}$,
\[
\prob{\san(D, \cQ)_{-\id} \in S} \leq e^{\eps} \prob{\san(D, \cQ')_{-\id} \in S} + \delta.
\]
\end{definition}

In our proofs of both differential privacy and analyst privacy, we will often
establish that for any $D \sim D'$, the two distributions $\mathcal{A}(D),
\mathcal{A}(D')$ are such that with probability at least $1-\delta$ over $y
\getsr M(D)$,
\[
  \left| \ln\left(\frac{\Pr[\mathcal{A}(D) = y]}{\Pr[\mathcal{A}(D') =
  y]}\right)\right| \leq \eps.
\]
This condition implies $(\eps, \delta)$-differential privacy~\cite{DRV10}.

\SUBSECTION{Queries and Accuracy}

In this work we consider two types of queries: \emph{low-sensitivity queries}
and \emph{linear queries}.  Low-sensitivity queries are parameterized by $\Delta
\in [0,1]$: a \emph{$\Delta$-sensitive query} is any function $q \from
\univ^n \to [0,1]$ such that
\begin{mathdisplayfull}
\max_{D \sim D'} |q(D) - q(D')| \leq \Delta.
\end{mathdisplayfull}%
A \emph{linear query} is a particular type of low-sensitivity query, specified
by a function $q\from \cX \to [0,1]$.  We define the evaluation of $q$
on a database $D \in \cX^n$ to be
\begin{mathdisplayfull}
q(D) = \frac{1}{n} \sum_{i=1}^{n} q(x^{(i)}),
\end{mathdisplayfull}%
so a linear query is evidently $(1/n)$-sensitive.

Since $\san$ may output a data structure, we must specify how to answer queries
in $\cQ$ from the output $\san(D)$.  Hence, we require that there is an
\emph{evaluator} $\cE\from \cR \times \cQ \to \R$ that estimates $q(D)$ from the
output of $\san(D)$.  For example, if $\san$ outputs a vector of ``noisy
answers'' $Z = \{ q(D) + Z_q | q \in \cQ \}$, where $Z_q$ is a random variable for
each $q \in \cQ$, then $\cR = \R^{\cQ}$ and $\cE(Z, q)$ is the $q$-th component
of $Z$.  Abusing notation, we write $q(Z)$ and $q(\san(D))$ as shorthand for
$\cE(Z, q)$ and $\cE(\san(D),q)$, respectively.
\begin{definition} 
  \label{def:acc}
  An output $Z$ of a mechanism $\san(D)$ is \emph{$\alpha$-accurate} for query
  set $\cQ$ if
  $
  |q(Z) - q(D)| \leq \alpha
  $
  for every $q \in \cQ$.  A mechanism is \emph{$(\alpha, \beta)$-accurate} for
  query set $\cQ$ if for every database $D$,
  \[
  \prob{\forall q \in \cQ, |q(\san(D)) - q(D)| \leq \alpha} \geq 1-\beta,
  \]
  where the probability is taken over the coins of $\san$.
\end{definition}

\SUBSECTION{Differential Privacy Tools}
We will use a few previously known differentially private mechanisms.  When we
need to answer a small number of queries we will use the well-known Laplace
mechanism~\cite{DMNS06}, with an improved analysis from~\cite{DRV10}.
\begin{lemma} 
  \label{lem:laplace}
  Let $\cF = \set{f_1, \dots, f_{|\cF|}}$ be a set of $\Delta$-sensitive queries
  $f_i \from \univ^n \to [0,1]$, and let $D \in \univ^n$ be a database. Let
  $\epsilon,\delta \leq 1$. Then the mechanism $\cA_{\mathrm{Lap}}(D, \cF)$ that
  outputs
  \begin{mathdisplayfull}
    f_i(D) + \Lap\left(\frac{\Delta \sqrt{8 |\cF| \log(1/\delta)}}{\eps}\right)
  \end{mathdisplayfull}
  for every $f_i \in \cF$ is:
  \begin{enumerate}
    \item $(\eps, \delta)$-differentially private, and
    \item $(\alpha, \beta)$-accurate for any $\beta \in (0,1]$ and \shortbreak
      $\alpha = \eps^{-1} \Delta \sqrt{8 |\cF| \log(1/\delta)}
      \log(|\cF|/\beta)$.
  \end{enumerate}
\end{lemma}
When we need to answer a large number of queries, we will use the multiplicative
weights mechanism from~\cite{HR10}, with an improved analysis from Gupta et
al.~\cite{GRU12}.
\begin{lemma} 
  \label{lem:mw}
  Let $\cF = \set{f_1, \dots, f_{|\cF|}}$ be a set of $(1/n)$-sensitive
  linear queries, $f_i \from \univ^n \to [0,1]$.  Let $D \in \univ^n$ be a database.
  Then there is a mechanism $\cA_{\mathrm{MW}}(D, \cF)$ that is:
  \begin{enumerate}
    \item $(\eps, \delta)$-differentially private, and
    \item $(\alpha, \beta)$-accurate for any $\beta \in (0,1]$ and \shortbreak
      \begin{mathdisplayfull}
        \alpha = O \left( \frac{ \log^{1/4} |\univ| \sqrt{ \log(|\cF|/\beta)
        \log(1/\delta)}}{\eps^{1/2} n^{1/2}} \right).
      \end{mathdisplayfull}
  \end{enumerate}
\end{lemma}
\begin{remark}
We use the above lemma as a black box, agnostic to the algorithm which instantiates these guarantees.
\end{remark}
Our algorithms also use the private sparse vector algorithm.  This algorithm
takes as input a database and a set of low-sensitivity queries, with the promise
that only a small number of the queries have large answers on the input
database.  Its output is a set of queries with large answers on the input
database. Importantly for this work, the sparse vector algorithm
(cf.~\cite{HR10,RothNotes}) ensures the privacy of the input queries in a strong
sense.
\begin{lemma} 
  \label{lem:sv}
  Let $\cF = \set{f_1, \dots, f_{|\cF|}}$ be a set of $\Delta$-sensitive
  functions, $f_i \from \univ^n \to [0,1]$.  Let $D \in \univ^n$ be a database,
  $\alpha \in (0,1]$, $k \in [|\cF|]$ such that
  \begin{mathdisplayfull}
    |\set{i \mid f_i(D) \geq \alpha}| \leq k.
  \end{mathdisplayfull}
  Then there is an algorithm $\cA_{\mathrm{SV}}(D, \cF)$ that
  \begin{enumerate}
    \item is $(\eps, \delta)$-differentially private with respect to $D$,
    \item returns $I \subseteq [|\cF|]$ of size at most $k$ such that with
      probability at least $1-\beta$,
\iffull
      \[
      \set{i \mid f_i(D) \geq \alpha  + \eps^{-1}\Delta \sqrt{8 k \log(1/\delta)} \log(|\cF|/\beta)}
      \subseteq I \subseteq
      \set{i \mid f_i(D) \geq \alpha  },
      \]
\fi
    \item and is perfectly private with respect to the queries: \\
      if $\cF' = \set{f_1, \dots, f'_j, \dots, f_k}$, then for every $D$ and $i
      \neq j$,
      \[
      \prob{i \in \cA_{\mathrm{SV}}(D, \cF)} = \prob{i \in \cA_{\mathrm{SV}}(D, \cF')}.
      \]
  \end{enumerate}
\end{lemma}

We will also use the Composition Theorem of Dwork, Rothblum, and
Vadhan~\cite{DRV10}.
\begin{lemma} 
\label{lem:composition}
Let $\cA \from \univ^* \to \cR^T$ be a mechanism such that for every pair of
adjacent inputs $x \sim x'$, every $t \in [T]$, every $r_{1}, \dots, r_{t-1} \in
\cR$, and every $r_t \in \cR$,
\begin{align*}
  &\prob{\cA_t(x) = r_{t} \mid \cA_{1, \dots, t-1}(x) = r_{1}, \dots, r_{t-1}}
  \\ &\leq e^{\eps_0} \prob{\cA_t(x') = r_{t} \mid \cA_{1, \dots, t-1}(x') = r_{1}, \dots, r_{t-1}} + \delta_0
\end{align*}
for $\eps_0 \leq 1/2$.
Then $\cA$ is $(\eps, \delta)$-differentially private for $\eps = \sqrt{8 T
\log(1/\delta)} + 2 \eps_0^2 T$ and $\delta = \delta_0 T$.
\end{lemma}

\SUBSECTION{Multiplicative Weights}
Let $A \from \cA \to [0,1]$ be a measure over a set of actions $\cA$.  We use
$|A| = \sum_{a \in \cA} A(a)$ to denote the {\em density} of $A$.  A measure
naturally corresponds to a probability distribution $\dist{A}$ in which
\begin{mathdisplayfull}
\prob{\dist{A} = a} = A(a)/|A|
\end{mathdisplayfull}
for every $a \in \cA$.  Throughout, we will use calligraphic letters $(\cA)$ to
denote a set of actions, lower case letters ($a$) to denote the actions, capital
letters ($A$) to denote a measure over actions, and capital letters with a tilde
to denote the corresponding distributions ($\dist{A}$).
\iffull
We will use the KL-divergence between two distributions, defined to be
\[
KL(\dist{A} || \dist{A}') = \sum_{a \in \cA} \dist{A}(a) \log\left(\dist{A}(a) /
\dist{A}'(a)\right).
\]
\fi
Let $L \from \cA \to [0,1]$ be a loss function (losses $L$).  Abusing notation,
we can define $L(A) = \ex{L(\dist{A})}.$  Given an initial measure $A_1$, we can
define the multiplicative weights algorithm in Algorithm~\ref{alg:MW}.
\begin{algorithm}[h!]
  \caption{The Multiplicative Weights Algorithm, $MW_{\eta}$}
  \label{alg:MW}
  \begin{algorithmic}
    \STATE{For $t = 1,2,\dots,T$:}
    \INDSTATE[1]{Sample $a_{t} \getsr \dist{A}_t$}
    \INDSTATE[1]{Receive losses $L_t$ (may depend on $A_{1},a_1, \dots, A_{t-1},
    a_{t-1}$)}
    \INDSTATE[1]{\textbf{Update:} {\bf For each} $a \in \cA${\bf :}}
    \INDSTATE[2]{Update $A_{t+1}(a) = e^{-\eta L_t(a)} A_{t}(a)$ for every $a \in \cA$}
  \end{algorithmic}
\end{algorithm}

\shortpagebreak

The following theorem about the multiplicative weights update is well-known.

\begin{theorem}[\iffull Multiplicative Weights.\fi See e.g. \cite{RakhlinNotes}]
  Let $A_{1}$ be the uniform measure of density $1$, and let $\set{a_1, \dots,
  a_T}$ be the actions obtained by $MW_{\eta}$ with  losses $\set{L_1, \dots,
  L_t}$.  Let $A^* = \mathbf{1}_{a = a^*}$, for some $a^* \in \cA$, and $\delta
  \in (0,1]$.  Then with probability at least $1-\beta$,
\iffull
  \begin{align*}
    \Ex{t \getsr [T]}{L_t(a_t)} &\leq{} (1+\eta) \Ex{t \getsr [T]}{L_t(A^*)} +
    \frac{KL(\dist{A}^* || \dist{A}_1)}{\eta T} + \frac{4
    \log(1/\beta)}{\sqrt{T}} \\
    &\leq{} \Ex{t \getsr [T]}{L_t(A^*)} + \eta + \frac{\log |\cA|}{\eta T} +
    \frac{4 \log(1/\beta)}{\sqrt{T}}.
  \end{align*}
\fi
\end{theorem}

We need to work with a variant of multiplicative weights that only produces
measures $A$ of high density, which will imply that $\dist{A}$ does not assign
too much probability to any single element of $\cA$.  To this end, we will apply
(a special case of) the Bregman projection to the measures obtained from the
multiplicative weights update rule.
\begin{definition} 
  \label{def:proj}
  Let $s \in (0,\cU]$.  Given a measure $A$ such that $|A| \leq s$, let
  $\projs{A}$ be the \emph{(Bregman) projection of $A$ into the set of
  density-$s$ measures}, obtained by computing $c \geq 1$ such that $s = \sum_{a
  \in \cA} \min\{1, c A(a)\}$ and setting $\proj{A}(a) = \min\{1, c M(a)\}$ for
  every $a \in \cA$. We call $s$ is the \emph{density} of measure $A$.
\end{definition}

\begin{algorithm}[h!]
  \caption{The Dense Multiplicative Weights Algorithm, $DMW_{s, \eta}$}
  \label{alg:DMW}
  \begin{algorithmic}
    \STATE{For $t = 1,2,\dots,T$:}
    \INDSTATE[1]{Let $A'_{t} = \projs{A}_t$, and sample $a_{t} \getsr
    \dist{A}'_{t}$}
    \INDSTATE[1]{Receive losses $L_t$ (may depend on $A_{1},a_1, \dots, A_{t-1},
    a_{t-1}$)}
    \INDSTATE[1]{\textbf{Update:} {\bf For each} $a \in \cA${\bf :}}
    \INDSTATE[2]{Update $A_{t+1}(a) = e^{-\eta L_t(a)} A_{t}(a)$}
  \end{algorithmic}
\end{algorithm}

Given an initial measure $A_1$ such that $|A_1| \leq s$, we can define the dense
multiplicative weights algorithm in Algorithm~\ref{alg:DMW}.  Note that we
update the unprojected measure $A_{t}$, but sample $a_{t}$ using the projected
measure $\projs{A}_{t}$.  Observe that the update step can only decrease the
density, so we will have $|A_{t}| \leq s$ for every $t$.  Like before,
given a sequence of losses $\set{L_1, \dots, L_T}$ and an initial measure
$A_{1}$ of density $s$, we can consider the sequence $\set{A_1, \dots, A_{T}}$
where $A_{t+1}$ is given by the projected multiplicative weights update applied
to $A_{t}, L_{t}$.  The following theorem is known.

\begin{theorem}
  Let $A_1$ be the uniform measure of density $1$ and let $\set{a_1, \dots, a_T}$
  be the sequence of measures obtained by $DMW_{s,\eta}$ with losses $\set{L_1,
  \dots, L_T}$.  Let $A^* = \mathbf{1}_{a \in S*}$ for some set $S^* \subseteq
  \cA$ of size $s$, and $\delta \in (0,1]$.  Then with probability $1-\beta$,
\iffull
  \begin{align*}
    \Ex{t \getsr [T]}{L_t(\proj{A_t)}} &\leq (1+\eta) \Ex{t \getsr [T]}{L_t(A^*)} +
    \frac{KL(\dist{A}^* || \dist{A}_1)}{\eta T} + \frac{4 \log(1/\beta)}{\sqrt{T}} \\
    &\leq{} \Ex{t \getsr [T]}{L_t(A^*)} + \eta + \frac{\log |\cA|}{\eta T} +
    \frac{4 \log(1/\beta)}{\sqrt{T}}.
  \end{align*}
\fi
\end{theorem}

\SUBSECTION{Regret Minimization and Two-Player Zero-Sum Games} Let $G \from
\cA_R \times \cA_C \to [0,1]$ be a two-player zero-sum game between players
($R$)ow and ($C$)olumn, who take actions $r \in \cA_{R}$ and $c \in \cA_{C}$ and
receive losses $G(r,c)$ and $-G(r,c)$, respectively. Let $\Delta(\cA_R),
\Delta(\cA_C)$ be the set of measures over actions in $\cA_{R}$ and $\cA_{C}$,
respectively.  The well-known minimax theorem states that
\[
  v := \min_{R \in \Delta(\cA_{R})} \max_{C \in \Delta(\cA_{C})} G(R, C) =
  \max_{C \in \Delta(\cA_{C})} \min_{R \in \Delta(\cA_{R})} G(R, C).
\]
We define this quantity $v$ to be the \emph{value of the game}.

Freund and Schapire~\cite{FS} showed that if two sequences of actions $\set{r_1,
\dots, r_T}, \set{c_1, \dots, c_T}$ are ``no-regret with respect to one
another'', then $\dist{r} = \frac{1}{T} \sum_{t=1}^{T} r_t$ and $\dist{c} =
\frac{1}{T} \sum_{t=1}^{T} c_t$ form an approximate equilibrium strategy pair.
More formally, if
\[
  \max_{c \in \cA_{C}} \Ex{t}{G(r_t, c)} - \rho \leq \Ex{t}{G(r_t, c_t)} \leq
  \min_{r \in \cA_{R}} \Ex{t}{G(r, c_t)} + \rho,
\]
then
\begin{mathdisplayfull}
v - 2\rho \leq G(\dist{r}, \dist{c}) \leq v + 2\rho.
\end{mathdisplayfull}%
Thus, if Row chooses actions using the multiplicative weights update
rule with losses $L_t(r_t) = G(r_t, c_t)$ and Column chooses actions using the
multiplicative weights rule with losses $L_t(r_t) = -G(r_t, c_t)$, then each
player's distribution on actions converges to a minimax strategy.  \iffull
That is, if we play until both players have regret at most $\rho$:
\[
\max_{c} G(\dist{r}, c) \leq v + 2\rho
\qquad
v - 2\rho \leq \min_{r} G(r, \dist{c}).
\]
\fi

For query privacy in our view of query release as a two player game, Column must
not put too much weight on any single query.  Thus, we need an analogue of this
result in the case where Column is not choosing actions according to the
multiplicative weights update, but rather using the projected multiplicative
weights update.  In this case we cannot hope to obtain an approximate minimax
strategy, since Column cannot play any single action with significant
probability.  However, we can define an alternative notion of the value of a
game where Column is restricted in this way: let $\Delta_{s}(\cA_{C})$ be the
set of measures over $\cA_{C}$ of minimum density at least $s$, and define
\[
v_{s} := \min_{R \in \Delta(\cA_{R})} \max_{C \in \Delta_{s}(\cA_{C})} G(R, C).
\]
Notice that $v_{s} \leq v$, and $v_{s}$ can be very different from $v$.
\begin{theorem}
  Let $\set{r_1, \dots, r_T} \in \cA_{R}$ be a sequence of row-player actions,
  $\set{C_1, \dots, C_T} \in \Delta_{s}(\cA_{C})$ be a sequence of high-density
  measures over column-player actions, and $\set{c_1, \dots, c_T} \in \cA_{C}$
  be a sequence of column-player actions such that $c_j \getsr C_j$ for every $t
  \in [T]$.  Further, suppose that
\iffull
  \[
    \Ex{t}{G(r_t, c_t)} \leq \min_{R \in \Delta(\cA_{R})} \Ex{t}{G(R, c_t)} + \rho
    \quad \textrm{and} \quad \Ex{t}{G(r_t, c_t)} \geq \max_{C \in
      \Delta_{s}(\cA_{C})} \Ex{t}{G(r_t, C)} - \rho.
  \]
\fi
  Then,
  \begin{mathdisplayfull}
    v_{s} - 2\rho \leq G(\dist{r}, \dist{c}) \leq v + 2\rho.
  \end{mathdisplayfull}%
  Moreover, $\dist{r}$ is an approximate min-max strategy with respect to
  strategies in $\Delta_{s}(\cA_{C})$, i.e.,
  \begin{mathdisplayfull}
    v_{s} - 2\rho \leq \max_{C \in \Delta_{s}(\cA_{C})} G(\dist{r}, C) \leq v +
    2\rho.
  \end{mathdisplayfull}
\end{theorem}
\iffull
\begin{proof}
  For the first set of inequalities, we handle each part separately. For one
  direction,
  \begin{align*}
    v_{s}
    &={} \min_{R \in \Delta(\cA_{R})} \max_{C \in \Delta_{s}(\cA_{C})} G(R,C) \\
    &\leq{} \max_{C \in \Delta_{s}(\cA_{C})} \Ex{t}{G(r_t, C)}
    \leq{} \Ex{t}{G(r_t, c_t)} + \rho \\
    &\leq{} \min_{R \in \Delta(\cA_{R})} \Ex{t}{G(R, c_t)} + 2\rho
    ={} \min_{R \in \Delta(\cA_{R})} G(R, \dist{c}) +2\rho \\
    &\leq{} G(\dist{r}, \dist{c}) + 2\rho.
  \end{align*}
  The other direction is similar, starting with the fact that $v = \max_{c \in
  C} \min_{r \in R} G(r,c)$.

For the second set of inequalities, we also handle the two cases separately.
For the upper bound,
\begin{align*}
\max_{C \in \Delta_{s}(\cA_{C})} \Ex{t}{G(\dist{r}, C)}
&\leq{} \Ex{t}{G(r_t,c_t)} + \rho \\
&\leq \min_{R \in \Delta(\cA_{R})} \Ex{t}{G(R, c_{t})} + 2\rho
= \min_{R \in \Delta(\cA_{R})} G(R, \dist{c}) + 2\rho \\
&\leq v + 2\rho.
\end{align*}

For the lower bound,
\begin{align*}
\max_{C \in \Delta_{s}(\cA_{C})} G(\dist{r}, C)
\geq{} \Ex{t}{G(\dist{r}, \dist{c})} \geq v_{s} - 2\rho
\end{align*}
This completes the proof of the theorem.
\end{proof}
\fi


\begin{corollary} \label{cor:approxminmax}
Let $G \from \cA_{R} \times \cA_{C} \to [0,1]$.  If the row player chooses
actions $\set{r_1, \dots, r_{T}}$ by running $MW_{\eta}$ with loss functions
$L_{t}(r) = G(r, c_t)$ and the column player chooses actions $\set{c_1, \dots,
  c_{T}}$ by running $DMW_{s, \eta}$ with the loss functions $L_{t}(c) = -G(r_t,
  c)$, then with probability at least $1-\beta$,
\begin{mathdisplayfull}
v_{s} - 2\rho \leq \max_{c \in C_{s}} G(\dist{r}, c) \leq v + 2\rho,
\end{mathdisplayfull}
for
\[
\rho = \eta + \frac{ \max\{\log |\cA_{R}|, \log |\cA_{C}|\}}{ \eta T} +
\frac{4\log(2/\beta)}{\sqrt{T}}.
\]
\end{corollary}

\section{A One-Query-To-Many-Analyst Private Mechanism}

\iffull
\subsection{An Offline Mechanism for Linear Queries}
\fi
We define our offline mechanisms for releasing linear queries in
Algorithm~\ref{alg:offline-q}.
\begin{algorithm}[h!]
  \caption{Offline Mechanism for Linear Queries with One-Query-to-Many-Analyst Privacy}
  \begin{algorithmic} \label{alg:offline-q}
    \STATE{\textbf{Input:} Database $D \in \cX^n$ and sets of linear queries $\cQ_1, \dots, \cQ_m$.}
    \STATE{\textbf{Initialize:}
    Let $\cQ = \bigcup_{j = 1}^{m} \cQ_j \cup \neg \cQ_j$,
    $D_0(x) = 1/|\univ|$ for every $x \in \univ$,
    $Q_0(q) =  1/|\cQ|$ for every $q \in \cQ$,
    $$T  = n \cdot \max\{\log|\univ|, \log |\cQ|\}, \longquad \eta =
    \frac{\eps}{2\sqrt{T \log(1/\delta)}}, \longquad s = 12 T$$}
    \STATE{\textbf{DataPlayer:}}
    \INDSTATE[1]{On input a query $\hat{q}_t$, for each $x \in \univ$:}
    \INDSTATE[2]{Update $D_{t}(x) = D_{t-1}(x) \cdot \exp\left(-\eta \left(\frac{1 + \hat{q}_t(D) - \hat{q}_t(x)}{2}\right)\right)$}
    \INDSTATE[1]{Choose $\hat{x}_{t} \getsr \dist{D}_{t}$ and send $\hat{x}_{t}$ to \textbf{QueryPlayer}}
    \STATE{}
    \STATE{\textbf{QueryPlayer:}}
    \INDSTATE[1]{On input a data element $\hat{x_t}$, for each $q \in \cQ$:}
    \INDSTATE[2]{Update $Q_{t+1}(q) = Q_{t}(q) \cdot
    \exp\left(-\eta\left(\frac{1 + q(D) -q(\hat{x}_t)}{2} \right)\right)$}
    \INDSTATE[2]{Let $P_{t+1} = \projs Q_{t+1}$}
    \INDSTATE[1]{Choose $\hat{q}_{t+1} \getsr \dist{P}_{t+1}$ and send $\hat{q}_{t+1}$ to \textbf{DataPlayer}}
    \STATE{}
    \STATE{\textbf{GenerateSynopsis:}}
    \INDSTATE[1]{Let $\widehat{D} = (\hat{x}_1, \dots, \hat{x}_{T})$.}
    \INDSTATE[1]{Run sparse vector on $\widehat{D}$, obtain a set of at most $s$ queries $\cQ_{f}$}
    \INDSTATE[1]{Run Laplace Mechanism, obtain answer $a_{q}$ for each $q \in \cQ_{f}$}
    \INDSTATE[1]{Output $\widehat{D}$ to all analysts.}
    \INDSTATE[1]{For each $q \in \cQ_{f}$, output $(q, a_{q})$ to the analyst that issued $q$.}
  \end{algorithmic}
\end{algorithm}

\iffull
\subsubsection{Accuracy Analysis}
\fi
\begin{theorem} \label{thm:offlineacc4counting}
The offline algorithm for linear queries is $(\alpha, \beta)$-accurate for
\[
\alpha = O\left(\frac{\sqrt{\log(|\cX| +
|\cQ|)\log(1/\delta)}\log(|\cQ|/\beta)}{\eps \sqrt{n}}\right).
\]
\end{theorem}
\begin{proof}
Observe that the algorithm is computing an approximate equilibrium of the game
$G_{D}(x, q) = \frac{1 + q(D) - q(x)}{2}$.  Let $v, v_{s}$ be the value and
constrained value of this game, respectively.  First, we pin down the quantities
$v$ and $v_{s}$.
\begin{claim} \label{clm:acc1}
  For every $D$, the value and constrained value of $G_{D}$ is $1/2$.
\end{claim}
\iffull
\begin{proof}[\longproof of Claim~\ref{clm:acc1}]
  It's clear that the value (and hence constrained value) is at most $1/2$,
  because
  \[
    \min_{x} \max_{q} \frac{1 + q(D) - q(x)}{2} \leq \max_{q} \frac{1 + q(D) -
    q(D)}{2} = \frac{1}{2}.
  \]
  Suppose we choose $x$ such that $(1 + q(D) - q(x))/2 < 1/2$ for some $q \in
  \cQ$.  Then, since the query $q' = 1-q$ is also in $\cQ$, $(1+q'(D) - q'(x))/2
  > 1/2$.  But then $\max_{q \in \cQ} (1+q(D) - q(x))/2 > 1/2$, so the value of
  the game is at least $1/2$.

  For the constrained value, suppose we choose $x$ such that $\Ex{q \getsr
  Q}{(1+q(D)-q(x))/2} < 1/2$ for some $Q \in Q_{s}$.  Then we can flip every query
  in $Q$ to get a new distribution $Q'$ such that $\Ex{q \getsr Q'}{(1 + q(D) -
  q(x))/2} > 1/2$.  So $v_{s} \geq 1/2$ as well.
\end{proof}
\fi

Let $\hat{D} = \frac{1}{T} \sum_{t=1}^{T} x_{t}$.  By
Corollary~\ref{cor:approxminmax},
\[
  v_{s} - 2\rho \leq \max_{Q \in \Delta_{s}(\cQ)} \left( \frac{1}{2} \Ex{q
    \getsr \dist{Q}}{1 + q(D) - q(\hat{D})}\right) \leq v + 2\rho.
\]
Applying Claim~\ref{clm:acc1} and rearranging terms, with probability at least
$1 - \beta/3$,
\iffull
\begin{align*}
&\left| \max_{Q \in \Delta_{s}(\cQ)} \left(\Ex{q \getsr \dist{Q}}{q(D) -
  q(\hat{D})}\right) \right|
=  \max_{Q \in \Delta_{s}(\cQ)} \left(\Ex{q \getsr \dist{Q}}{\left|q(D) - q(\hat{D})\right|}\right)
\leq{} 4\rho \\
&={} 4\left( \eta + \frac{\max\{\log|\cX|, \log|\cQ|\}}{\eta T} + \frac{4 \log(2/\beta)}{\sqrt{T}} \right) \\
&={}  O\left(\frac{\sqrt{\log(|\cX| + |\cQ|) \log(1/\delta)} +
\log(1/\beta)}{\eps \sqrt{n}}\right) := \alpha_{\hat{D}}.
\end{align*}
\fi

The previous statement suffices to show that $|q(D) - q(\dist{D})| \leq
\alpha_{\hat{D}}$ for all but $s$ queries.  Otherwise, the uniform distribution
over the bad queries would be a distribution over queries contained in
$\Delta_{s}(\cQ)$, with expected error larger than $\alpha_{\hat{D}}$.

We can now run the sparse vector algorithm (Lemma~\ref{lem:laplace}). With
probability at least $1-\beta/3$, it will identify every query $q$ with error
larger than $\alpha_{\hat{D}} + \alpha_{SV}$ for
\[
\alpha_{\mathrm{SV}} =
O\left(\frac{\sqrt{s \log(1/\delta)}\log(|\cQ|/\beta)}{\eps n} \right).
\]
Since there are at most $s$ such queries, with probability at least $1-\beta/3$,
the Laplace mechanism (Lemma~\ref{lem:sv}) answers these queries to within error
\[
\alpha_{\mathrm{Lap}} = O\left( \frac{\sqrt{s \log(1/\delta)}\log(s/\beta)}{\eps
n} \right).
\]
Now, observe that in the final output, there are two ways that a query can be
answered: either by $\hat{D}$, in which case its answer can have error as large
as $\alpha_{\hat{D}} + \alpha_{\mathrm{SV}}$, or by the Laplace mechanism, in
which case its answer can have error as large as $\alpha_{\mathrm{Lap}}$.  Thus,
with probability at least $1-\beta$, every query has error at most
$\max\{\alpha_{\hat{D}} + \alpha_{\mathrm{SV}}, \alpha_{\mathrm{Lap}}\}$.
Substituting our choice of $s = 12T = O(n \log(|\univ| + |\cQ|))$ and
simplifying, we conclude that the mechanism is $(\alpha, \beta)$-accurate for
\[
\alpha = O\left(\frac{\sqrt{\log(|\cX| + |\cQ|)
\log(1/\delta)}\log(|\cQ|/\beta)}{\eps \sqrt{n}} \right).
\]
\end{proof}

\iffull
\subsubsection{Data Privacy}
\fi

\begin{theorem} \label{thm:dataprivacy4counting}
  Algorithm~\ref{alg:offline-q} satisfies $(\eps, \delta)$-differential privacy
  for the data.
\end{theorem}

Before proving the theorem, we will state a useful lemma about the Bregman
projection onto the set of high density measures (Definition~\ref{def:proj}).
\begin{lemma} [Projection Preserves \iffull Differential \fi Privacy] \label{lem:projdp}
  Let $A_0, A_1 \from \cA \to [0,1]$ be two full-support measures over a set of
  actions $\cA$ and $s \in (0,|\cA|)$ be such that $|A_0|, |A_1| \leq s$ and
  $|\ln(A_0(a) / A_1(a))| \leq \eps$ for every $a \in \cA$.  Let $A'_0 =
  \projs{A_0}$ and $A'_1 = \projs{A_1}$.  Then $|\ln(A'_0(a) / A'_1(a))| \leq
  2\eps$ for every $a \in \cA$.
\end{lemma}
\iffull
\begin{proof} [\longproof of Lemma~\ref{lem:projdp}]
  Recall that to compute $A' = \projs{A}$, we find a ``scaling factor'' $c > 1$
  such that
  \[
    \sum_{a \in \cA} \min \{1, c A(a)\} = s,
  \]
  and set $A'(a) = \min \{1, c A(a) \}$.  Let $c_0$ and $c_1$ be the scaling
  factors for $A'_0$ and $A'_1$ respectively.  Assume without loss of generality
  that $c_0 \leq c_1$.  First, observe that
  \[
    \left| \ln \left( \frac{\min\{1, c_0 A_0(a)\}}{\min\{1, c_0 A_1(a)\}}
    \right) \right| \leq \left| \ln \left( \frac{A_0(a)}{ A_1(a)} \right)
    \right| \leq \eps,
  \]
  for every $a \in \cA$.  Second, we observe that $c_1 / c_0 \leq e^{\eps}$.  If
  this were not the case, then we would have $c_1 A_1(a) \geq c_0 A_1(a) e^{\eps}
  \geq c_0 A_0(a)$ for every $a \in \cA$, with strict inequality for at least
  one $a$.  But then,
  \[
    \sum_{a \in \cA} \min \{1, c_1 A_1(a) \} > \sum_{a \in \cA} \min \{1, c_0
    A_0(a)\} = s,
  \]
  which would contradict the choice of $c_1$.  Thus,
  \begin{equation*}
    \left| \ln\left(\frac{\min \{1, c_0 A_0(a)\}}{\min \{1, c_1 A_1(a)\}}
    \right) \right| \leq \left| \ln\left( \frac{\min \{1, c_0 A_0(a)\}}{\min
      \{1, c_0 A_1(a)\}} \right)\right| + \left| \ln \left( \frac{c_1}{c_0}
      \right) \right| \leq \eps + \eps = 2\eps,
  \end{equation*}
  for every $a \in \cA$.
\end{proof}
\fi
Now we prove the main result of this section.
\begin{proof} [\longproof of Theorem~\ref{thm:dataprivacy4counting}]
  We focus on analyzing the privacy properties of the output $\hat{D} =
  (\hat{x}_{1}, \dots, \hat{x}_{T})$, the privacy of the final stage of the
  mechanism will follow from standard arguments in differential privacy.  We
  will actually show the stronger guarantee that the sequence $v = (\hat{x}_{1},
  \hat{q}_{1}, \dots, \hat{x}_{T}, \hat{q}_{T})$ is differentially private for
  the data.  Fix a pair of adjacent databases $D_0 \sim D_1$ and let $V_0, V_1$
  denote the distribution on sequences $v$ when the mechanism is run on database
  $D_0, D_1$ respectively.  We will show that with probability at least
  $1-\delta/3$ over $v = (\hat{x}_{1}, \hat{q}_{1}, \dots, \hat{x}_{T},
  \hat{q}_{T}) \getsr V_0$,
  \begin{mathdisplayfull}
    \left| \ln\left(\fullfrac{V_0(v)}{V_1(v)} \right) \right| \leq
    \fullfrac{\eps}{3},
  \end{mathdisplayfull}
  which is no weaker than $(\eps/3, \delta/3)$-differential privacy.  To do so,
  we analyze the privacy of each element of $v$, $\hat{x}_t$ or $\hat{q}_{t}$,
  and apply the composition analysis of Dwork, Rothblum, and
  Vadhan~\cite{DRV10}. Define $\eps_0 = 2\eta T / n$.
  \begin{claim} \label{clm:nextx} For every $v$, and every $t \in [T]$,
    \[
      \left| \ln\left(\frac{V_0(\hat{x}_t \mid \hat{x}_1, \hat{q}_1, \dots,
      \hat{x}_{t-1}, \hat{q}_{t-1})}{V_1(\hat{x}_t \mid \hat{x}_1, \hat{q}_1,
      \dots, \hat{x}_{t-1}, \hat{q}_{t-1})}\right) \right| \leq \eps_0.
    \]
  \end{claim}
  \begin{proof}[\longproof of Claim~\ref{clm:nextx}]
\iffull
    We can prove the statement by the following direct calculation.
    \begin{align*}
      &\left| \ln\left(\frac{V_0(\hat{x}_t \mid \hat{x}_1, \hat{q}_1, \dots,
      \hat{x}_{t-1}, \hat{q}_{t-1})}{V_1(\hat{x}_t \mid \hat{x}_1, \hat{q}_1,
      \dots, \hat{x}_{t-1}, \hat{q}_{t-1})}\right) \right|
      ={} \left| \ln\left (\frac{\exp\left(-(\eta/2) \sum_{j=1}^{t-1} 1 +
      \hat{q}_j(D_0) - \hat{q}_j(\hat{x}_t) \right)}{\exp\left(-(\eta/2)
        \sum_{j=1}^{t-1} 1 + \hat{q}_j(D_1) - \hat{q}_j(\hat{x}_t) \right)}
        \right) \right| \\
      ={} &\frac{\eta}{2} \left| \left ( \sum_{j=1}^{t-1} 1 + \hat{q}_j(D_0) -
      \hat{q}_j(\hat{x}_t) \right) - \left(\sum_{j=1}^{t-1} 1 + \hat{q}_j(D_1) -
      \hat{q}_j(\hat{x}_t) \right) \right| \\
      ={} &\frac{\eta}{2} \left| \sum_{j=1}^{t-1} \hat{q}_j(D_0) -
      \hat{q}_j(D_1) \right| \leq \frac{\eta (t-1)}{2 n} \leq \frac{\eta T}{2 n}
      \leq \eps_0
    \end{align*}
\fi
  \end{proof}
  \begin{claim} \label{clm:nextq} For every $v$, and every $t \in [T]$,
    \[
      \left| \ln\left(\frac{V_0(\hat{q}_t \mid \hat{x}_1, \hat{q}_1, \dots,
      \hat{x}_{t})}{V_1(\hat{q}_t \mid \hat{x}_1, \hat{q}_1, \dots,
      \hat{x}_{t})}\right) \right| \leq \eps_0.
    \]
  \end{claim}
  \begin{proof}[\longproof of Claim~\ref{clm:nextq}]
    The sample $\hat{q}_t$ is made according to $\dist{P}_{t}$, which is the
    distribution corresponding to the projected measure $P_{t}$.  First we'll
    look at the unprojected measure $Q_{t}$. Observe that, for any database $D$
    and query $q$,
    \[
      Q_{t}(q) = \exp\left( -(\eta/2) \sum_{j=1}^{t-1} 1 + q(D) - q(\hat{x}_j)\right).
    \]
    Thus, if $Q_0(q)$ is the measure we would have when database $D_0$ is the input,
    and $Q_1(q)$ is the measure we would have when database $D_1$ is the input, then
    \begin{equation*}
      \left| \ln \left(\frac{Q_0(q)}{Q_1(q)} \right) \right| \leq \frac{\eta}{2}
      \left| \sum_{j=1}^{t-1} q_j(D_0) - q_j(D_1) \right| \leq \frac{\eta
      T}{2n},
    \end{equation*}
    for every $q \in \cQ$. Given that $Q_0$ and $Q_1$ satisfy this condition,
    Lemma~\ref{lem:projdp} guarantees that the projected measures satisfy
    \[
      \left| \ln \left(\fullfrac{P_0(q)}{P_1(q)} \right) \right| \leq \fullfrac{\eta T}{n}.
    \]
    Finally, we note that if the above condition is satisfied for every $q \in
    \cQ$, then the distributions $\dist{P_0}, \dist{P_{1}}$ satisfy
    \begin{equation*}
      \left| \ln \left(\fullfrac{\dist{P}_0(q)}{\dist{P}_1(q)} \right) \right|
      \leq \fullfrac{2 \eta T}{n} \leq \eps_0,
    \end{equation*}
    because the value of the normalizer also changes by at most a multiplicative
    factor of $e^{\pm \eta T / n}$.  We observe that $V_b(\hat{q}_t \mid
    \hat{x}_1, \hat{q}_1, \dots, \hat{x}_{t}) = \dist{P}_b(\hat{q}_t)$ for $b
    \in \bits$, which completes the proof of the claim.
\end{proof}
Now, the composition lemma (Lemma~\ref{lem:composition}) (for
$2T$-fold composition) guarantees that with probability at least $1-\delta/3$,
\[
\left| \ln\left(\fullfrac{V_0(v)}{V_1(v)} \right) \right| \leq \eps_0 \sqrt{4T
\log(3/\delta)} + 4\eps_0^2 T,
\]
which is at most $\eps/3$ by our choice of $\eps_0$.  This implies that
$\hat{D}$ is $(\eps/3, \delta/3)$-differentially private.

We note that the sparse vector computation
to find the $s$ queries with large error is $(\eps/3, \delta/3)$-differentially
private, by our choice of parameters (Lemma~\ref{lem:sv}), and the answers to
the queries found by sparse vector are $(\eps/3, \delta/3)$-differentially
private for our choice of parameters
(Lemma~\ref{lem:laplace}).\iffull\footnote{We could improve the constants in our
  privacy analysis slightly by finding the queries with large error using sparse
  vector and answering them using the Laplace mechanism in one step.  However,
  in our algorithm for achieving analyst-to-many privacy, we need to do the
  analogous steps separately, and thus we chose to present them this way to
maintain modularity.} \fi \space The theorem follows from composition.
\end{proof}

\iffull
\subsubsection{Query Privacy}
\fi
\begin{theorem} \label{thm:queryprivacy4counting}
  Algorithm~\ref{alg:offline-q} satisfies $(\eps,
  \delta)$-one-query-to-many-analyst differential privacy.
\end{theorem}

Before proving query privacy of Algorithm~\ref{alg:offline-q}, we will state a
useful composition lemma.  The lemma is a generalization of the ``secrecy of the
sample lemma''~\cite{KLNRS07, DRV10} to the interactive setting. Consider the
following game:
\begin{itemize}
  \item Fix an $(\eps, \delta)$-differentially private mechanism $\cA \from \cU^*
    \to \cR$ and a bit $b \in \bits$.  Let $D_0 = \emptyset$.
  \item For $t = 1, \dots , T$:
    \begin{itemize}
      \item The (randomized) adversary $\cB(y_1,\dots,y_{t}; r)$
        chooses two distributions $B^0_t, B^1_t$ such that $SD(B^0_t, B^1_t)
        \leq \sigma$.
      \item Choose $x_{t} \getsr B^b_{t}$ and let $D_t = D_{t-1} \cup
        \set{x_{t}}$.
      \item Choose $y_{t} \getsr \cA(D_{t})$.
    \end{itemize}
\end{itemize}
For a fixed mechanism $\cA$ and adversary $\cB$, let $V^0$ be the distribution
on $(y_1, \dots, y_{T})$ when $b = 0$ and $V^1$ be the distribution on $(y_1,
\dots, y_{T})$ when $b = 1$.
\begin{lemma} \label{lem:sdtodp}
  If $\eps \leq 1/2$ and $T \sigma \leq 1/12$, then with probability at least $1
  - T\delta - \delta'$ over $y = (y_1, \dots, y_T) \getsr V^0$,
  \[
    \left| \ln \left( \fullfrac{V^0(y)}{V^1(y)} \right) \right| \leq \eps (T\sigma)
    \sqrt{2T \log(1/\delta')} + 30\eps^2(T\sigma) T.
  \]
\end{lemma}
\iffull
(We prove this lemma in Appendix~\ref{sec:app1}.)
\fi

We also need another lemma about the Bregman projection onto the set of
high-density measures (Definition~\ref{def:proj})
\begin{lemma} 
  \label{lem:projsd}
  Let $A_0 \from \cA \to [0,1]$ and $A_1 \from \cA \cup \set{a^*} \to [0,1]$ be
  two full-support measures over their respective sets of actions and $s \in
  (0,|\cA|)$ be such that 1) $|A_0|, |A_1| \leq s$ and 2) $A_0(a) = A_1(a)$ for
  every $a \in \cA$.  Let $A'_0 = \projs{A_0}$ and $A'_1 = \projs{A_1}$.  Then
  $SD(\dist{A}'_0, \dist{A}'_1) \leq 1/s$.
\end{lemma}
\iffull
\begin{proof} [\longproof of Lemma~\ref{lem:projsd}]
  Using the form of the projection (Definition~\ref{def:proj}), it is not hard
  to see that for $a \neq a^*$, $A'_0(a) \geq A'_1(a)$.  For convenience, we
  will write $A'_0(a^*) = 0$ even though $a^*$ is technically outside of the
  domain of $A'_0$.  We can now show the following.
  \begin{align*}
    \sum_{a \in \cA \cup \set{a^*}} | A'_0(a) - A'_1(a) |
    &= |A'_{0}(a^*) - A'_1(a^*)| + \sum_{a \neq a^*} | A'_0(a) - A'_1(a) | \\
    &\leq{} 1 + \sum_{a \neq a^*} | A'_{0}(a) - A'_1(a) | \\
    &={} 1 + \sum_{a \neq a^*} A'_{0}(a) - A'_{1}(a) \tag{$A'_{0}(a) \geq
    A'_{1}(a)$ for $a \neq a^*$} \\
    &={} 1 + |A'_{0}| - (|A'_{1}| - A'_{1}(a^*))
    \leq{} 1 + |A'_{0}| - (|A'_{1}| - 1) \\
    &={} 1 + s - (s - 1) = 2
  \end{align*}
  We also have that $|A'_{0}| = |A'_{1}| = s$, so
  \begin{align*}
    SD(\dist{A}'_{0}, \dist{A}'_{1})
    &={} \frac{1}{2} \sum_{a \in \cA \cup \set{a^*}} \left|
    \frac{A'_{0}(a)}{|A'_{0}|} - \frac{A'_{1}(a)}{|A'_{1}|} \right| \\
    &={} \frac{1}{2s} \sum_{a \in \cA \cup \set{a^*}} | A'_{0}(a) - A'_{1}(a) |
    \leq \frac{1}{s}.
  \end{align*}
\end{proof}
\fi
Now we can prove one-query-to-many-analyst privacy.

\begin{proof} [\longproof of Theorem~\ref{thm:queryprivacy4counting}]
Fix a database $D$.  Consider two adjacent query sets $\cQ_0 \sim \cQ_1$ and,
without loss of generality assume $\cQ_0 = \cQ_1 \cup \{q^*\}$ and that $q^* \in
\cQ_{\id}$ for some analyst $\id$.  We write the output to all analysts as $v =
(\hat{x}_1, \dots, \hat{x}_T, b_{1}, \dots, b_{|\cQ|}, a_{1}, \dots, a_{|\cQ|})$
where $\hat{D} = \{\hat{x}_{1}, \dots, \hat{x}_{T}\}$ is the database that is
released to all analysts, $b_{1}, \dots, b_{|\cQ|}$ is a sequence of bits that
indicates whether or not $q_j(\hat{D})$ is close to $q_j(D)$, and $a_{1}, \dots,
a_{|\cQ|}$ is a sequence of approximate answers to the queries $q_j(D)$ (or
$\bot$, if $q_j(\hat{D})$ is already accurate).  We write $v_{-\id}$ for the
portion of $v$ that excludes outputs specific to analyst $\id$'s queries.  Let
$V_0, V_1$ be the distribution on outputs when the query set is $\cQ_0$ and
$\cQ_1$, respectively.

We analyze the three parts of $v$ separately.  First we show that
$\hat{D}$, which is shared among all analysts, satisfies analyst privacy.
\begin{claim} \label{clm:analyst1}
  With probability at least $1-\delta$ over the samples $\hat{x}_1, \dots,
  \hat{x}_{T} \getsr V_0$,
  \[
    \left| \ln\left( \frac{V_0(\hat{x}_1, \dots, \hat{x}_{T})}{V_1(\hat{x}_1,
    \dots, \hat{x}_{T})} \right) \right| \leq \eps.
  \]
\end{claim}
\begin{proof}[\longproof of Claim~\ref{clm:analyst1}]
  To prove the claim, we show how the output $\hat{x}_1, \dots, \hat{x}_T$ can
  be viewed as the output of an instantiation of the mechanism analyzed by
  Lemma~\ref{lem:sdtodp}. For every $t \in [T]$ and $\hat{q}_1, \dots,
  \hat{q}_{t-1}$, we define the measure $D_{t}$ over database items to be
  \[
    D_{t}(x) = \exp\left(-(\eta/2) \sum_{j=1}^{t-1} 1 + \hat{q}_j(D) -
    \hat{q}_j(x) \right).
  \]
  Notice that if we replace a single query $\hat{q}_\ell$ with $\hat{q}_{\ell}'$
  and obtain the measure $D'_{t}$, then for every $x \in \cX$,
  \begin{mathdisplayfull}
    \left| \ln \left( \fullfrac{\dist{D}_{t}(x)}{\dist{D}'_{t}(x)} \right) \right|
    \leq \eta.
  \end{mathdisplayfull}
  Thus we can view $\hat{x}_{t}$ as the output of an $\eta$-differentially
  private mechanism $\cA_{D}(\hat{q}_1, \dots, \hat{q}_{t-1})$, which fits into
  the framework of Lemma~\ref{lem:sdtodp}. (Here, $\hat{x}_{t}$ plays the role
  of $y_t$ and $\hat{q}_1, \dots, \hat{q}_{t-1}$ plays the role of $D_{t-1}$ in
  the description of the game, while the input database $D$ is part of the
  description of $\cA$).

  Now, in order to apply Lemma~\ref{lem:sdtodp}, we need to argue the
  distribution on samples $\hat{q}_{t}$ when the query set is $\cQ_0$ is
  \emph{statistically close} to the distribution on samples $\hat{q}_{t}$ when
  the query set is $\cQ_1$.  Fix any $t \in [T]$ and let $Q_{0}, Q_{1}$ be the
  measure $Q_{t}$ over queries maintained by the query player when
  the input query set is $\cQ_0, \cQ_1$, respectively.  For $q \neq q^*$, we
  have
  \[
    Q_{0}(q) = Q_{1}(q) = \exp\left(-(\eta/2) \sum_{j=1}^{t-1}1+ q(D) -
    q(\hat{x}_j) \right).
  \]
  Additionally, we set $Q_{0}(q^*) = 0$ (for notational convenience), while
  $Q_{1}(q^*) \in (0,1]$.  Thus, if we let $P_{0} = \projs{Q_{0}}$ and $P_{1} =
  \projs{Q_{1}}$, we will have $SD(\dist{P}_{0}, \dist{P}_1) \leq 1/s$ by
  Lemma~\ref{lem:projsd}. Since the statistical distance is $1/s = 1/12T$, we
  can apply Lemma~\ref{lem:sdtodp} to show that with probability at least
  $1-\delta$,
  \begin{align*}
    \left| \ln\left( \frac{V(\hat{x}_1, \dots, \hat{x}_{T})}{V'(\hat{x}_1,
    \dots, \hat{x}_{T})} \right) \right|
    \leq{} &\frac{\eta \sqrt{T \log(1/\delta)}}{8} + \frac{5 \eta^2 T}{2}
    \leq{} \eps. {\iffull \tag{$\eta = \eps / (2 \sqrt{T \log(1/\delta)})$} \fi}
  \end{align*}
\end{proof}

Now that we have shown $\hat{D}$ satisfies
$(\eps,\delta)$-one-query-to-many-analyst differential privacy, it remains to
show that the remainder of the output satisfies perfect
one-query-to-many-analyst privacy.  Recall from the proof of
Theorem~\ref{thm:offlineacc4counting} that $\hat{D}$ will be accurate for all
but $s$ queries.  That is, if we let $\set{f_j}_{j \in [|\cQ|]}$ consist of the
functions $f_j(D) = | q_j(D) - q_j(\hat{D}) |$, then
\begin{mathdisplayfull}
| \set{j \mid f_j(D) \geq \alpha} | \leq s,
\end{mathdisplayfull}%
where $\alpha$ is chosen as in Theorem~\ref{thm:offlineacc4counting}.  By
Lemma~\ref{lem:sv}, the sparse vector algorithm will release bits $b_1, \dots,
b_{|\cQ|}$ (the indicator vector of the subset  of queries with large error)
such that for every $j \in [|\cQ|]$, the distribution on $b_j$ does not depend
on any function $f_{j'}$ for $j' \neq j$.  Thus, if $z_{-a}$ contains all the
bits of $b_1, \dots, b_{|\cQ|}$ that do not correspond to queries in $\cQ_{a}$,
then the distribution of $z_{-\id}$ does not depend on the queries asked by
analyst $\id$, and thus $z_{-\id}$ is perfectly one-query-to-many analyst
private.  Finally, for each query $q_j$ such that $b_j = 1$, the output to the
owner of that query will include $a_j = q_j(D) + z_j$ where $z_j$ is an
independent sample from the Laplace distribution. These outputs do not depend
on any other query, and thus are perfectly one-query-to-many analyst private.
This completes the proof of the theorem.

\end{proof}

\section{A One-Analyst-to-Many-Analyst Private Mechanism}
\iffull
\subsection{An Offline Mechanism for Linear Queries}
\fi
In this section we present an algorithm for answering linear queries that
satisfies the stronger notion of one-analyst-to-many-analyst privacy.  The
algorithm is similar to Algorithm~\ref{alg:offline-q}, but with two notable
modifications.

First, instead of the ``query player'' of Algorithm~\ref{alg:offline-q}, we will
have an ``analyst player'' who chooses analysts as actions and is trying to find
an analyst $\id \in [m]$ for which there is at least one query in $\cQ_{\id}$
with large error (recall that the queries are given to the mechanism in sets
$\cQ_1, \dots, \cQ_m$).  That is, the analyst player attempts to find $\id \in
[m]$ to maximize $\max_{q \in \cQ_{\id}} q(D) - q(\hat{D})$.

Second, we will compute a database $\hat{D}$ such that $\max_{q \in \cQ_{\id}}
|q(D) - q(\hat{D})|$ is small for all but $s$ \emph{analysts} in the set $[m]$,
rather than having the $s$ mishandled queries in Algorithm~\ref{alg:offline-q}.
We can still use sparse vector to find these $s$ analysts, however we can't
answer the queries with the Laplace mechanism, since each of the analysts may
ask an exponential number of queries. However, since there are not too many
analysts remaining, we can use $s$ independent copies of the multiplicative
weights mechanism (each run with $\eps' \approx \eps/\sqrt{s}$) to handle each
analyst's queries.

\begin{algorithm}[h!]
  \caption{Offline Mechanism for Linear Queries with One-Analyst-to-Many-Analyst Privacy}
  \begin{algorithmic} \label{alg:offline-id}
    \STATE{\textbf{Input:} Database $D \in \cX^n$, and $m$ sets of linear
    queries $\overline{\cQ}_{1}, \dots, \overline{\cQ}_{m}$.  For $\id \in [m]$,
    let $\cQ_{\id} = \overline{\cQ}_{\id} \cup \neg \overline{\cQ}_{\id}$.}
    \STATE{\textbf{Initialize:} Let $D_0(x) = 1/|\univ|$ for each $x \in \univ$,
    $I_0(q) =  1/m$ for each $\id \in [m]$, \\
    \[
      T  = n^{2/3} \max\{\log|\univ|, m\}, \longquad \eta = \frac{\sqrt{T
      \log(1/\delta)}}{2 \epsilon}, \longquad s = 12 T.
    \]
   }
    \STATE{\textbf{DataPlayer:}}
    \INDSTATE[1]{On input an analyst $\hat{\id}_t$, for each $x \in \univ$,
  update:}
    \[
      D_{t}(x) = D_{t-1}(x) \cdot
    \exp\left(-\eta \max_{q \in \cQ_{\hat{\id}_{t}}} \left(\frac{1 +
      \hat{q}_t(D) - \hat{q}_t(x)}{2}\right)\right)
    \]
    \INDSTATE[1]{Choose $\hat{x}_{t} \getsr \dist{D}_{t}$ and send $\hat{x}_{t}$
    to \textbf{AnalystPlayer}}
    \STATE{}
    \STATE{\textbf{AnalystPlayer:}}
    \INDSTATE[1]{On input a data element $\hat{x_t}$, for each $\id \in \cI$,
  update:}
    \[
      I_{t+1}(\id) = I_{t}(\id) \cdot
    \exp\left(-\eta \max_{q \in \cQ_{\id}} \left(\frac{1 + q(D)
      -q(\hat{x}_t)}{2} \right)\right)
    \]
    \INDSTATE[1]{Let $P_{t+1} = \projs I_{t+1}$}
    \INDSTATE[1]{Choose $\hat{\id}_{t+1} \getsr \dist{P}_{t+1}$ and send
    $\hat{\id}_{t+1}$ to \textbf{DataPlayer}}
    \STATE{}
    \STATE{\textbf{GenerateSynopsis:}}
    \INDSTATE[1]{Let $\widehat{D} = (\hat{x}_1, \dots, \hat{x}_{T})$}
    \INDSTATE[1]{Run sparse vector on $\widehat{D}$, obtain a set of at most $s$
  analysts:}
    \INDSTATE[2]{$I_{f} = \set{\id_1, \dots, \id_s} \subseteq [m]$}
    \INDSTATE[1]{For each analyst $\id \in I_{f}$, run $\cA_{\mathrm{MW}}(D,
    \cQ_{\id})$ with parameters}
    \INDSTATE[2]{$\eps' = \frac{\eps}{10\sqrt{s \log(3s/\delta)}}$ and $\delta' =
    \frac{\delta}{3s}$}
    \INDSTATE[2]{Obtain a sequence of answers $\vec{a}_{\id}$.}
    \INDSTATE[1]{Output $\widehat{D}$ to all analysts.}
    \INDSTATE[1]{For each $\id \in [m] \setminus I_{f}$, output $\vec{a}_{\id}$
  to analyst $\id$}
  \end{algorithmic}
\end{algorithm}

\iffull
\subsubsection{Accuracy Analysis}
\fi
\begin{theorem} \label{thm:offlineacc4counting2}
  Algorithm~\ref{alg:offline-id} is $(\alpha, \beta)$-accurate for
  \[
    \alpha = \tilde{O}\left( \frac{\sqrt{\log(|\univ| + m)\log|\cQ_{\id}|}
    \log(m/\beta) \log^{3/4}(1/\delta)}{\eps n^{1/3}} \right).
  \]
\end{theorem}
\iffull
\begin{proof}
  As we discussed above,  the algorithm is computing an approximate equilibrium
  of the game
  \[
    G_{D, \cQ}(x, \id) = \max_{\id \in [m]} \max_{q \in \cQ_{\id}}\frac{1 + q(D)
    - q(x)}{2}.
  \]
  Let $v, v_{s}$ be the value and constrained value of this game, respectively.
  First we pin down the quantities $v$ and $v_{s}$.
  \begin{claim} \label{clm:acc2}
    For every $D, m, \cQ$, the value and constrained value of $G_{D, m, \cQ}$ is
    $1/2$.
  \end{claim}
  The proof of this claim is omitted, but is nearly identical to that of
  Claim~\ref{clm:acc1}.

  Let $\hat{D} = \frac{1}{T} \sum_{t=1}^{T} \hat{x}_{t}$.  By
  Corollary~\ref{cor:approxminmax},
  \[
    v_{s} - 2\rho \leq \max_{I \in \Delta_{s}([m])} \Ex{\id \getsr
      \dist{I}}{\max_{q \in \cQ_{\id}}\left( \frac{1 + q(D) -
        q(\hat{D})}{2}\right)} \leq v + 2\rho.
  \]
  Applying Claim~\ref{clm:acc2} and rearranging terms, we have that with
  probability $1-\beta/3$,
  \begin{align*}
    &\left| \max_{I \in \Delta_{s}([m])} \left(\Ex{\id \getsr \dist{I}}{\max_{q \in \cQ_{\id}} q(D) - q(\hat{D})}\right) \right|
    =  \max_{I \in \Delta_{s}([m])} \left(\Ex{\id \getsr \dist{I}}{\max_{q \in \cQ_{\id}} \left|q(D) - q(\hat{D})\right|}\right)
    \leq{} 4\rho \\
    &={} 4\left( \eta + \frac{\max\{\log|\cX|, \log m\}}{\eta T} + \frac{4 \log(3/\beta)}{\sqrt{T}} \right) \\
    &={}  O\left(\frac{\sqrt{\log(|\cX| + m) \log(1/\delta)} +
  \log(1/\beta)}{\eps n^{1/3}}\right) := \alpha_{\hat{D}}.
  \end{align*}
  The previous statement suffices to show that $\max_{q \in \cQ_{\id}} |q(D) -
  q(\dist{D})| \leq \alpha_{\hat{D}}$ for all but $s$ analysts $\id \in [m]$.
  Otherwise, the uniform distribution over the analysts for which the error
  bound of $\alpha_{\hat{D}}$ does not hold would be a distribution over
  analysts, contained in $\Delta_{s}([m])$ with expected error larger than
  $\alpha_{\hat{D}}$.

  Since there are at most $s$ such analysts we can run the sparse vector algorithm
  (Lemma~\ref{lem:sv}), and, with probability at least $1-\beta/3$, it will
  identify every analyst $\id$ such that the maximum error over all queries in
  $\cQ_{\id}$ is larger than $\alpha_{\hat{D}} + \alpha_{SV}$ for
  \[
    \alpha_{\mathrm{SV}} = O\left(\frac{\sqrt{s
    \log(1/\delta)}\log(m/\beta)}{\eps n} \right).
  \]
  There are at most $s$ such analysts.  Thus, running the multiplicative weights
  mechanism (Lemma~\ref{lem:mw}) independently for each of these analysts'
  queries---with privacy parameters $\eps' =
  \Theta(\eps/\sqrt{s\log(s/\delta)})$ and $\delta' = \Theta(\delta/s)$---will
  yield answers such that, with probability $1-\beta/3$, for every $\id \in I'$,
  \begin{align*}
    \max_{q \in \cQ_{\id}} |q(D) - a_{q}|
    &\leq O\left( \frac{s^{1/4} \log^{1/4}|\univ| \sqrt{\log(s |\cQ_{\id}|/\beta)}\log^{3/4}(s/\delta)}{\sqrt{\eps n}} \right) \\
    &\leq  \tilde{O}\left( \frac{n^{1/6} \sqrt{\log(|\univ| + m) \log(|\cQ_{\id}|/\beta)}\log^{3/4}(1/\delta)}{\sqrt{\eps n}} \right) \\
    &\leq  \tilde{O}\left( \frac{ \sqrt{\log(|\univ| + m)
      \log(|\cQ_{\id}|/\beta)}\log^{3/4}(1/\delta)}{n^{1/3} \sqrt{\eps}} \right)
      := \alpha_{\mathrm{MW}}.
  \end{align*}
  Taking a union bound, observing that the maximum error on any query is
  $\max\{\alpha_{\hat{D}} + \alpha_{\mathrm{SV}}, \alpha_{\mathrm{MW}}\}$, and
  simplifying, we get that the mechanism is $(\alpha,\beta)$-accurate for
  \[
    \alpha = \tilde{O}\left( \frac{\sqrt{\log(|\univ| + m)\log|\cQ_{\id}|}
    \log(m/\beta) \log^{3/4}(1/\delta)}{\eps n^{1/3}} \right).
  \]
\end{proof}
\fi

\iffull
\subsubsection{Data Privacy}
\fi

\begin{theorem} \label{thm:dataprivacy4counting2}
  Algorithm~\ref{alg:offline-id} satisfies $(\eps, \delta)$-differential privacy
  for the data.
\end{theorem}

\iffull
We omit the proof of this theorem, which follows that of
Theorem~\ref{thm:dataprivacy4counting} almost identically.  The only difference
is that in the final step, we need to argue that running $s$ independent copies
of multiplicative weights with privacy parameters $\eps' = \Theta(\eps / \sqrt{s
\log(s/\delta)})$ and $\delta' = \Theta(\delta/s)$ satisfies $(\eps/3,
\delta/3)$-differential privacy, which follows directly from the composition
properties of differential privacy (Lemma~\ref{lem:composition}).
\fi

\iffull
\subsubsection{Query Privacy}
In this section we prove query privacy for our one analyst to many analyst
mechanism.
\fi
\begin{theorem} \label{thm:queryprivacy4counting2}
  Algorithm~\ref{alg:offline-id} satisfies $(\eps,
  \delta)$-one-analyst-to-many-analyst differential privacy.
\end{theorem}
\iffull
\begin{proof}
  Fix a database $D$.  Consider two adjacent sets of queries $\cQ_0, \cQ_1$.
  Without loss of generality assume $\cQ_0 = \cQ_{\id_{1}} \cup \dots
  \cQ_{\id_{m}}$ and $\cQ_1 = \cQ_0 \cup \cQ_{\id^*}$. That is $\cQ_1$ is just
  $\cQ_0$ with an additional set of queries $\cQ_{\id^*}$ added.  We write the
  output to all analysts as $v = (\hat{x}_1, \dots, \hat{x}_T, b_{1}, \dots,
  b_{m}, \vec{a}_{1}, \dots, \vec{a}_{m})$ where $\hat{D} = \hat{x}_{1}, \dots,
  \hat{x}_{T}$ is the database that is released to all analysts, $b_{1}, \dots,
  b_{m}$ is a sequence of bits that indicates whether or not $q_j(\hat{D})$ is
  close to $q_j(D)$ for every $q \in \cQ_{\id}$, and $\vec{a}_{1}, \dots,
  \vec{a}_{m}$ is a sequence consisting of the output of the multiplicative
  weights mechanism for every analyst $\id \in [m]$ and $\bot$ for every other
  analyst.  Let $V_0, V_1$ be the distribution on outputs when the queries are
  $\cQ_0$ and $\cQ_1$, respectively.

  The proof closely follows the proof of one-query-to-many-analyst privacy for
  Algorithm 3.  Showing that the final two parts $b, \vec{a}$ of the output are
  query private is essentially the same, so we will focus on proving that
  $\hat{D}$ satisfies one-analyst-to-many-analyst privacy.
  \begin{claim} \label{clm:analyst2}
    With probability at least $1-\delta$ over $\hat{x}_{1}, \dots, \hat{x}_{T}
    \getsr V_0$,
    \[
      \left| \ln\left( \frac{V_0(\hat{x}_1, \dots, \hat{x}_{T})}{V_1(\hat{x}_1,
      \dots, \hat{x}_{T})} \right) \right| \leq \eps.
    \]
  \end{claim}
  \begin{proof}[\longproof of Claim~\ref{clm:analyst2}]
    To prove the claim, we show how the output $\hat{x}_1, \dots, \hat{x}_T$ can
    be viewed as the output of an instantiation of the mechanism analyzed by
    Lemma~\ref{lem:sdtodp}.  Notice that for every $t \in [T]$ and $\hat{\id}_1,
    \dots, \hat{\id}_{t-1}$, we can write the measure $D_{t}$ over database
    items as
    \[
      D_{t}(x) = \exp\left(-(\eta/2) \sum_{j=1}^{t-1}\max_{q \in
        \cQ_{\hat{\id}_{j}}} 1 + \hat{q}_j(D) - \hat{q}_j(x) \right).
    \]
    If we replace a single analyst $\hat{\id}_\ell$ with $\hat{\id}_{\ell}'$,
    and obtain the measure $D'_{t}$, then for every $x \in \cX$,
    \[
      \left| \ln \left( \frac{\dist{D}_{t}(x)}{\dist{D}'_{t}(x)} \right) \right|
      \leq \eta.
    \]
    Thus we can view $\hat{x}_{t}$ as the output of an $\eta$-differentially
    private mechanism $\cA_{D}(\hat{\id}_1, \dots, \hat{\id}_{t-1})$, which fits
    into the framework of Lemma~\ref{lem:sdtodp}. (Here, $\hat{x}_{t}$ plays the
    role of $y_t$ and $\hat{\id}_1, \dots, \hat{\id}_{t-1}$ plays the role of
    $D_{t-1}$ in the description of the game, while the input database $D$ is
    part of the description of $\cA$).

    As before, we apply Lemma~\ref{lem:sdtodp}, to argue that the distribution
    on analysts $\hat{\id}_{t}$ when the query set is $\cQ_0$ is statistically
    close to the distribution on analysts $\hat{\id}_{t}$ when the analyst set
    is $\cQ_1$.  The argument does not change significantly, thus we can apply
    Lemma~\ref{lem:sdtodp} to show that with probability at least $1-\delta$,
    \begin{align*}
      \left| \ln\left( \frac{V(\hat{x}_1, \dots, \hat{x}_{T})}{V'(\hat{x}_1,
      \dots, \hat{x}_{T})} \right) \right|
      \leq{} &\frac{\eta \sqrt{T \log(1/\delta)}}{8} + \frac{5 \eta^2 T}{2}
      \leq{} \eps \tag{$\eta = \eps / (2 \sqrt{T \log(1/\delta)})$}.
    \end{align*}
  \end{proof}

As before, the remainder of the output satisfies perfect
one-analyst-to-many-analyst privacy. This completes the proof of the theorem.
\end{proof}
\fi

\section{A One-Query-to-Many-Analyst \shortbreak Private Online Mechanism}
In this section, we present a mechanism that provides one-query-to-many-analyst
privacy in an online setting.  The mechanism can give accurate answers to any
\emph{fixed} sequence of queries that are given to the mechanism one at a time,
rather than the typical setting of \emph{adaptively chosen} queries.

The mechanism is similar to the online multiplicative weights algorithm of Hardt
and Rothblum~\cite{HR10}.  In their algorithm, a hypothesis about the true
database is maintained throughout the sequence of queries.  When a query
arrives, it is classified according to whether or not the current hypothesis
accurately answers that query.  If it does, then the query is answered according
to the hypothesis.  Otherwise, the query is answered with a noisy answer
computed from the true database and the hypothesis is updated using the
multiplicative weights update rule.

The main challenge in making that algorithm query private is to argue that the
hypothesis does not depend too much on the previous queries.  We overcome this
difficulty by ``sampling from the hypothesis.'' (recall that a database can be
thought of as a distribution over the data universe).  We must balance the need
to take many samples --- so that the database we obtain by sampling accurately
reflects the hypothesis database, and the need to limit the impact of any one
query on the sampled database.  To handle both these constraints, we introduce
\emph{batching} --- instead of updating every time we find a query not
well-answered by the hypothesis, we batch together $s$ queries at a time, and do
one update on the average of these queries to limit the influence of any single
query.

\begin{algorithm}[h!]
  \caption{Analyst-Private Multiplicative Weights for Linear Queries}
  \begin{algorithmic}
    \STATE{\textbf{Input:}\!\! Database $D \in \univ^n$, sequence $q_1, \dots, q_k$
  of linear queries}
  \STATE{\textbf{Initialize:}
  $D_0(x)= 1/|\univ|$ for each $x \in \univ$,
  $H_{0} = D_{0}$, $\cU_0 = \emptyset$,
  $s_0 = s + \Lap(2/\eps)$,
  $t = 0$, $r = 0,$
  \[
    \eta = \frac{1}{n^{2/5}}, \longquad s = \frac{128 n^{2/5} \sqrt{\log |\univ|
    \log(4k/\beta) \log(1/\delta)}}{\eps},
  \]
  \[
    \hat{n} = 32 n^{4/5} \log (4k / \beta), \longquad T = n^{4/5} \log |\univ|,
    \longquad R = 2sT,
  \]
  \[
    \sigma = \frac{20000 \log^{3/4} |\univ| \log^{1/4}(4k/\beta)
    \log^{5/4}(4/\delta)}{\eps^{3/2} n^{2/5}},
  \]
  \[
    \tau = \frac{80000 \log^{3/4} |\univ|
    \log^{5/4}(4k/\beta)\log^{5/4}(4/\delta)}{\eps^{3/2} n^{2/5}}.
  \]
}
\STATE{\textbf{AnswerQueries:}}
\STATE{{\bf While} $t < T, r < R, i \leq k$, on input query $q_i${\bf :}}
\INDSTATE[1]{Let $z_i = \Lap(\sigma)$}
\INDSTATE[1]{{\bf If} $|q_i(D) - q_i(H_t) + z_i| \leq \tau${\bf :} Output $q_i(H_t)$}
\INDSTATE[1]{{\bf Else:}}
\INDSTATE[2]{Let $u = \mathrm{sgn}(q_i(H_t) - q_i(D) - z_i) \cdot q_i$, $\cU_t = \cU_t \cup \set{u}$}
\INDSTATE[2]{Output $q_i(D) + z_i$}
\INDSTATE[2]{Let $r = r+1$}
\INDSTATE[2]{{\bf If} $|\cU_t| > s_t${\bf :}}
\INDSTATE[3]{Let $(D_{t+1}, H_{t+1}) = \mathbf{Update}(D_t, \cU_t)$}
\INDSTATE[3]{Let $\cU_{t+1} = \emptyset$, $s_{t+1} = s + \Lap(2/\eps)$}
\INDSTATE[3]{Let $t = t+1$}
\INDSTATE[1]{Advance to query $q_{i+1}$}
\STATE{}
\STATE{\textbf{Update:}}
\INDSTATE[1]{\textbf{Input:} distribution $D_t$, update queries $\cU_t = \set{u_1, \dots, u_{s_t}}$}
\INDSTATE[1]{{\bf For each} $x \in \univ${\bf :}}
\INDSTATE[2]{Let $\mathbf{u}_{t}(x) = \frac{1}{3s} \sum_{j=1}^{s_t} u_j(x)$}
\INDSTATE[2]{Update $D_{t+1}(x) = \exp (-(\alpha'/2) \mathbf{u}_{t}(x)) D_{t}(x)$}
\INDSTATE[1]{Normalize $D_{t+1}$}
\INDSTATE[1]{Let $H_{t+1}$ be $\hat{n}$ independent samples from $D_{t+1}$}
\INDSTATE[1]{\textbf{Return:} $(D_{t+1}, H_{t+1})$}
\end{algorithmic}
\label{alg:online}
\end{algorithm}

A note on terminology: the execution of the algorithm takes place in several
rounds, where each round processes one query.  Rounds where the query is
answered using the real database are called \emph{bad rounds}; rounds that are
not bad are \emph{good rounds}.  We will split the rounds into $T$
\emph{epochs}, where the hypothesis $H_t$ is used during epoch $t$.

\SUBSECTION{Accuracy}
First, we sketch a proof that the online mechanism answers linear queries
accurately. Intuitively, there are three ways that our algorithm might give an
inaccurate answer, and we treat each separately.  First, in a good round, the
answer given by the hypothesis may be a bad approximation to the true answer.
Second, in a bad round, the answer given may have too much noise.  We address
these two cases with straightforward arguments showing that the noise is not too
large in any round.

The third way the algorithm may be inaccurate is if there are more than $R$ bad
rounds, and the algorithm terminates early.  We show that this is not the case
using a potential argument: after sufficiently many bad rounds, the hypothesis
$D_{T}$ and the sample $H_{T}$ will be accurate for all queries in the stream,
and thus there will be no more bad rounds.  The potential argument is a simple
extension of the argument in Hardt and Rothblum~\cite{HR10} that handles the
additional error coming from taking samples from $D_{t}$ to obtain $H_{t}$.

\iffull
We will use the following tail bound on sums of Laplace variables.
\begin{lemma}[\cite{GRU12}]
  \label{lem:laplacetail}
  Let $X_1, \dots, X_{T}$ be $T$ independent draws from $\Lap(2/\eps)$, and let $X = \sum_{t=1}^{T}
  X_t$. Then,
  \[
    \Pr\left[\, |X| >  \frac{5\sqrt{T} \log(2/\beta) }{\eps} \right] < \beta.
  \]
\end{lemma}
\fi

\begin{theorem}
  Algorithm~\ref{alg:online} is $(\alpha, \beta)$-accurate for
  \[
    \alpha = O\left( \frac{\log^{3/2}(k/\beta) \sqrt{\log |\univ|
    \log(1/\delta)}}{\eps^{3/2} n^{2/5}} \right).
  \]
\end{theorem}
\iffull
We note that we can achieve a slightly better dependence on $k, |\univ|,
\frac{1}{\eps}, \frac{1}{\delta}, \frac{1}{\beta}$ by setting the parameters a
bit more carefully.  See Section~\ref{sec:onlinelowsens} for an intuitive
picture of how to set the parameters optimally.  We have made no attempt to
optimize the constant factors in the algorithm.
\begin{proof}
  First we show that, as long as the algorithm has not terminated early, it
  answers every query accurately.
  \begin{claim} \label{clm:smallnoise}
    Before the algorithm terminates, with probability $1-\beta/4$, every query
    is answered with error at most $\tau + 6\sigma \log(3k/\beta)$
  \end{claim}
  \begin{proof} [\longproof of Claim~\ref{clm:smallnoise}]
    Condition on the event that $|z_i| \leq 6 \sigma \log(4k/\beta)$ for every
    $i = 1,2,\dots,k$.  A standard analysis of the tails of the Laplace
    distribution shows that this event occurs with probability at least
    $1-\beta/3$.

    First we consider bad rounds.  In these rounds $q_i$ is answered with
    $q_i(D) + z_i$. Since we have assumed $|z_i|$ is not too large, all of these
    queries are answered accurately.

    Now we consider good rounds.  In these rounds we answer with $q_i(H_t)$, and
    we will only have a good round if $|q_i(D) - q_i(H_t) + z_i| \leq \tau$.
    Since we have assumed a bound on $|z_i|$, we can only have a good round if
    $|q_i(D) - q_i(H_t)| \leq \tau + 6 \sigma \log(3k/\beta)$.
\end{proof}

Now we must show that the algorithm does not terminate early.  Recall that it
can terminate early either because it hits a limit on the number of epochs, or
because it hits a limit on the number of bad rounds.  We will use a potential
argument to show that there cannot be too many epochs.  The number of bad rounds
that is in epoch $t$ is a random variable $s_{t}$, and we will also show that
with high probability, there are not too many bad rounds within the $T$ epochs.
\begin{claim} \label{clm:dontstop}
  With probability $1-3\beta/4$, the algorithm does not terminate before answering $k$ queries.
\end{claim}
\begin{proof}[\longproof of Claim~\ref{clm:dontstop}]
  We will use a potential argument a la Hardt and Rothblum~\cite{HR10} on the
  sequence of databases $D_{t}$.  The potential function will be $$\Phi_{t} =
  RE(D || D_{t}) := \sum_{x \in \univ} D(x) \log(D(x) / D_t(x)).$$  Elementary
  properties of the relative entropy function show that $\Phi_{t} \geq 0$ and
  $\Phi_{0} = RE(D || D_{0}) \leq \log |\univ|$.  A lemma of Hardt and Rothblum
  expresses the potential decrease from the multiplicative weights update rule
  in terms of the error of the current hypothesis on the update query.
  \begin{lemma}[\cite{HR10}] \label{lem:potential}
    \begin{mathdisplayfull}
      \Phi_{t-1} - \Phi_{t} \geq \eta \left( \mathbf{u}_{t}(D) -
      \mathbf{u}_{t}(D_{t-1}) \right) - \eta^2 / 4.
    \end{mathdisplayfull}
  \end{lemma}
  Since the potential function is bounded between $0$ and $\log |\univ|$, we can
  get a bound on the number of epochs by showing that the potential decreases
  significantly between most epochs.  Given the preceding lemma, we simply need
  to show that the queries $\mathbf{u}_{1}, \mathbf{u}_{2}, \dots$ have large
  (positive) error.

  Recall that $\mathbf{u}_{t} = \frac{1}{s} \sum_{u \in \cU_{t}} u$.  Also recall
  that if $u \in \cU$ and $u = q_i$, then the reason $q_i$ is in $\cU$ is because
  $q_i(D) - q_i(H_{t-1}) + z_i > \tau$.  Similarly, if $u = \neg q_i$, then
  $q_i(D) - q_i(H_{t-1}) + z_i < -\tau$.  We will focus on the first case where
  $q_i(D) - q_i(H_{t-1}) + z_i > \tau$, the other case will follow similarly.  We
  can get a lower bound on $u(D) - u(D_{t-1})$ as follows.
  \begin{align*}
    u(D) - u(D_{t-1})
    &\geq u(D) - u(H_{t-1}) + z_i - |z_i| - |q_i(H_{t-1}) - q_i(D_{t-1})| \\
    &\geq \tau - |z_i| - |q_i(H_{t-1}) - q_i(D_{t-1}) |
  \end{align*}
  We need to show that the right-hand side of the final expression is large.  We
  have already conditioned on the event that $|z_i| \leq 6 \sigma \log(3k/\beta)
  \leq \tau/4$.  Recall that $H_{t-1}$ is a collection of $\hat{n}$ samples from
  $D_{t-1}$.  Thus a simple Chernoff bound (over the $\hat{n}$ samples) and a
  union bound (over the $k$ queries) shows that, with probability $1-\beta/4$,
  for every $i \in [k]$, $|q_i(D_{t-1}) - q_i(H_{t-1})| \leq \sqrt{16 \log(3 T /
    \beta) / \hat{n}} \leq \tau/4$.

  Thus, with probability at least $1-2\beta/3$, for every $t$ and every $u \in
  \cU_t$,
  \[
    u(D) - u(D_{t-1}) \geq \tau - \tau/4 - \tau/4 = \tau/2.
  \]
  Now,
  \[
    \mathbf{u}_{t}(D) - \mathbf{u}_{t}(D_{t-1}) = \frac{1}{s} \sum_{u \in
      \cU_{t}} u(D) - u(D_{t-1}) \geq \frac{\tau |\cU_{t}|}{2s}.
  \]
  Conditioning on the event that all of the noise values $z_i$ are small and
  all of the sampled hypotheses $H_{t}$ are accurate for $D_{t}$ on every
  query, we can calculate
  \begin{align*}
    \Phi_{T'} - \Phi_{0}
    &\geq \frac{\eta \tau}{2s} \left(\sum_{t \leq T'} |\cU_{t}|\right) - T' \eta^2 \\
    &= \frac{\eta \tau}{2s} \left(\sum_{t \leq T'} s_{t} \right) - T' \eta^2
    =  \frac{\eta \tau T'}{2} + \frac{\eta \tau}{2s} \left(\sum_{t \leq T'}
    S_{t} \right) - T' \eta^2,
  \end{align*}
  where $S_{t}$ is the value of the sample $\Lap(2/\eps)$ used to compute
  $s_{t}$ in the $t$-th epoch.  Thus, applying Lemma~\ref{lem:laplacetail} to $S
  = \sum_{t \leq T'} S_{t}$, we have that with probability $1-\beta/4$,
  \begin{align*}
    \Phi_{T'} - \Phi_{0}
    &\geq  \frac{\eta \tau T'}{2} - \frac{\eta \tau}{2s} \left| S \right| - T' \eta^2 \\
    &\geq \frac{\eta \tau T'}{2} - \frac{\eta \tau}{2s} 5 \eps^{-1} \sqrt{T'}
    \log(20/\beta) - T' \eta^2.
  \end{align*}
  Now, noting that $\tau/2 > 8\eta$ and simplifying,
  \[
    \Phi_{T'} - \Phi_{0} \geq  2\eta^2 \tau T' - \eta^2 T' \geq \eta^2 T'.
  \]
  Thus, conditioning on all the events above, $T' \leq \log |\univ| / \eta^2
  \leq n^{4/5} \log |\univ|$.  These events all occur together with probability
  at least $1-3\beta/4$, and thus the algorithm does not terminate because it
  hits the limit of $T$ epochs.  Lastly, we need to show that the algorithm does
  not hit the limit of $R$ bad rounds within those at-most $T$ epochs.  Notice
  that the number of bad rounds is at most
  \begin{mathdisplayfull}
    \sum_{t=1}^{T} s_t = \sum_{t=1}^{T} s + S_{t},
  \end{mathdisplayfull}
  where $S_{t}$ is the sample of $\Lap(2/\eps)$ used to compute $s_{t}$.
  Applying Lemma~\ref{lem:laplacetail} again we have
  \begin{align*}
    \sum_{t=1}^{T} s + S_{t} \leq Ts + \sum_{t=1}^{T} S_{t} \leq Ts + 5
    \eps^{-1} \sqrt{T} \log(2/\beta) \leq 2Ts = R.
  \end{align*}
  Thus the algorithm does not terminate due to having more than $R$ bad rounds.
  Since the algorithm does not hit its limit of $T$ epochs or $R$ bad rounds,
  except with probability at most $3\beta/4$, the claim is proven.
\end{proof}

Combining the previous two claims proves the theorem.
\end{proof}
\fi

\SUBSECTION{Data Privacy}
\iffull
In this section we establish that our mechanism satisfies differential privacy.
Our proof will rely on a modular analysis of interactive differentially private
algorithms from Gupta, Roth, and Ullman~\cite{GRU12}.  Although we have not
presented our algorithm in their framework, the algorithm can easily be seen to
fit, and thus we will state an adapted version of their theorem without proof.

\begin{theorem}[\cite{GRU12}, Adapted] \label{thm:IDC}
  If Algorithm~\ref{alg:online} experiences at most $R$ bad rounds, and the
  parameters are set so that $\sigma \geq \frac{1000 \sqrt{R} \log(4/\delta)}{\eps
  n}$, then Algorithm~\ref{alg:online} is $(\eps, \delta)$-differentially private.
\end{theorem}

\begin{theorem}
  Algorithm~\ref{alg:online} satisfies $(\eps, \delta)$-differential privacy.
\end{theorem}
\begin{proof}
  The theorem follows directly from Theorem~\ref{thm:IDC} and our choice of $R$.
\end{proof}
\fi

\SUBSECTION{Query Privacy}
More interestingly, we show that this mechanism satisfies
one-query-to-many-analyst privacy.

\begin{theorem} \label{thm:qp}
  Algorithm~\ref{alg:online} is $(\epsilon, \delta)$-one-query-to-many-analyst
  private.
\end{theorem}
\begin{proof}
  Fix the input database $D$ and the coins of the Laplace noise --- we will show
  that for every value of the Laplace random variables, the mechanism satisfies
  analyst privacy. Consider any two adjacent sequences of queries $\cQ_0,
  \cQ_1$.  Without loss of generality, we will assume that $\cQ = q_1, \dots,
  q_k$ and $\cQ' = q^*, q_1, \dots, q_k$.  For notational simplicity, we assume
  that every query in $\cQ$ has a fixed index, regardless of the presence of
  $q^*$.  More generally, we could identify each query in the sequence by a
  unique index (say, a long random string) that is independent of the other
  queries. We want to argue that the answers to \emph{all queries in $\cQ$} are
  private, but \emph{not} that the answer to $q^*$ is private (if it is
  requested).
  
  We will represent the answers to the queries in $\cQ$ by a sequence
  $\set{(H_t, i_t)}_{t \in [T]}$ where $H_{t}$ is the hypothesis used in the
  $t$-th epoch and $i_{t}$ is the index of the last query in that epoch (the one
  that caused the mechanism to switch to hypothesis $H_{t}$).  Observe that for
  a fixed database $D$, Laplace noise, and sequence of queries $\cQ$, we can
  simulate the output of the mechanism \emph{for all queries in $\cQ$} given
  only this information --- once we fix a hypothesis $H_t$, we can determine
  whether any query $q$ will be added to the update pool in this epoch.  So once
  we begin epoch $t$ with hypothesis $H_{t}$, we have fixed all the bad rounds,
  and once we are given $i_{t}$, we have determined when epoch $t$ ends and
  epoch $t+1$ begins.  At this point, we fix the next hypothesis $H_{t+1}$ and
  continue simulating.

  Formally, let $V_0, V_1$ be distribution over sequences $\{(H_t, i_t)\}$ when
  the query sequence is $\cQ_0, \cQ_1$, respectively.  We will show that with
  probability at least $1-\delta$, if $\set{(H_t, i_t)}_{t \in [T]}$ is drawn
  from $V_0$, then
  \begin{mathdisplayfull}
    \left| \ln\left(\fullfrac{V_0(\set{(H_t, i_t)})}{V_1(\set{(H_t, i_t)})}
    \right) \right| \leq \eps.
  \end{mathdisplayfull}

  Recall that $\cU_{t}$ is the set of queries that are used to update the
  distribution $D_{t}$ to $D_{t+1}$.  We will use $\cU_{\leq t} = \bigcup_{j =
  0}^{t} \cU_{t}$ to denote the set of all queries used to update the
  distributions $D_{0}, \dots, D_{t}$.  Notice that if $q^*$ does not get added
  to the set $\cU_{0}$, then $V_0$ and $V_1$ will be distributed identically.
  Therefore, suppose $q^* \in \cU_{0}$.  First we must reason about the joint
  distribution of the first component of the output.
  \begin{claim} \label{clm:onlinequeryprv0}
    For all $H_0, i_0$,
    \begin{mathdisplayfull}
      \left| \ln \left( \fullfrac{V_0(H_0,i_0)}{V_1(H_0, i_0)} \right) \right|
      \leq \fullfrac{\eps}{2}.
    \end{mathdisplayfull}
  \end{claim}
\iffull
  \begin{proof}[\longproof of Claim~\ref{clm:onlinequeryprv0}]
    Since $H_0$ does not depend on the query sequence, it will be identically
    distributed in both cases.  Once $H_0$ is fixed, we can determine whether a
    query $q$ will cause an update.  Fix query $q_{i_0}$ and assume that it is
    the $s$-th update query in the sequence $q_1, \dots, q_k$ and  the
    $(s+1)$-st update query in the sequence $q^*, q_1, \dots, q_k$.  Then
    $V_0(i_0 | H_0) = \prob{s_0 = s}$ and $V_1(i_0 | H_0) = \prob{s_0 = s+1}$.
    By the basic properties of the Laplace distribution, $| \ln (V_0(i_0 | H_0)
    / V_1(i_0 | H_0)) | \leq \eps/100$.
  \end{proof}
\fi
  Now we reason about the remaining components $(H_1, i_1), \dots, (H_T, i_T)$.
  \begin{claim} \label{clm:onlinequeryprv1}
    For every $H_0, i_0$, with probability at least $1-\delta$ over the choice
    of components $v = (H_1, i_1, \dots, H_T, i_T) \getsr (V_0 \mid v_{t-1})$,
    we have
    \begin{mathdisplayfull}
      \left|\ln \left( \fullfrac{V_0(v \mid H_0, i_0)}{V_1(v \mid H_0, i_0)}
      \right) \right| \leq  \fullfrac{\eps}{2}.
    \end{mathdisplayfull}
  \end{claim}
\iffull
  \begin{proof}[\longproof of Claim~\ref{clm:onlinequeryprv1}]
    We will show that $v$ is the $\hat{n}T$-fold composition of $(\eps_0,
    0)$-differentially private mechanisms for suitable $\eps_0$.  Fix a prefix
    $v_{t-1} = H_0, i_0, \dots, H_{t-1}, i_{t-1}$.  Given this prefix, we can
    determine for any given sequence of queries $q_{1}, \dots, q_{i_{t-1}}$ or
    $q^*, q_1, \dots, q_{i_{t-1}}$ which queries are in the update set.
    Moreover, if $\cU_{< t}$ is the set of all update queries from the first
    query sequence, and $\cU'_{< t}$ is the set of all update queries from the
    second sequence, then $\cU_{< t} \triangle \cU'_{< t} = q^*$.

    Now consider the distribution of $H_t$.  Each sample in $H_t$ comes from the
    distribution $D_t$, which is either $$ D_{t}(x) \propto \exp\left( -( \eta/
    s) \sum_{u \in \cU_{< t}} u \right) \qquad \textrm{or} \qquad D'_{t}(x)
    \propto \exp\left( -(\eta/s) \frac{1}{s} \sum_{u \in \cU'_{< t}} u \right)
    $$ Given this, it is easy to see that for any $x$ we have $| \ln( D_{t}(x) /
    D'_{t}(x) ) | \leq 2\eta / s := \eps_{0}$.  Notice that once $i_{t-1}$ and
    $H_{t}$ are fixed, $i_{t}$ depends only on the choice of $s_t$ (the number
    of bad rounds to allow before updating the hypothesis), which is independent
    of the query sequence and thus incurs no additional privacy loss.  Thus the
    only privacy loss comes from the $\hat{n}$ samples in each of the $T$
    epochs, and the mechanism is a $\hat{n}T$-fold adaptive composition of
    $(\eps_0, 0)$ differentially private mechanisms.  A standard composition
    analysis (Lemma~\ref{lem:composition}) shows that the components $v$ are
    $(\eps', \delta)$-DP for $\eps' = \eps_0 \sqrt{2 \hat{n} T \log(1/\delta)} +
    2\eps_0^2 T \leq \eps/2$.  This completes the proof of the claim.
  \end{proof}
\fi
  Combining these two claims proves the theorem.
\end{proof}

\SUBSECTION{Handling Arbitrary Low-Sensitivity Queries} \label{sec:onlinelowsens}
We can also modify this mechanism to answer arbitrary $\Delta$-sensitive
queries, albeit with worse accuracy bounds.  As with our offline algorithms, we
modify the algorithm to run the multiplicative weights updates over the set of
databases $\univ^n$ and adjust the parameters.  When we run multiplicative
weights over a support of size $|\univ|^n$ (rather than $|\univ|$), the number
of epochs increases by a factor of $n$, which in turn affects the amount of
noise we have to add to ensure privacy.
\iffull

We will now sketch the argument, ignoring the parameters $\beta$ and $\delta$
for simplicity.  In order to get convergence of the multiplicative weights
distribution, we need to take $T \approx \frac{n \log |\univ|}{\eta^2}$ and in
order to ensure that $H_{t}$ approximates $D_{t}$ sufficiently well, we take
$\hat{n} \approx \sqrt{\frac{\log k}{\eta^2}}$.  Recall that to argue
\emph{analyst privacy}, we viewed the mechanism as being (essentially) the
$\hat{n} T$-fold composition of $\eps_0$-analyst private mechanisms, where
$\eps_0 = \eta/s$.  In order to get analyst privacy, we needed
\begin{align*}
  \frac{\eta}{s} \lesssim \frac{\eps}{ \sqrt{\hat{n}T}} \approx \frac{\eps
  \eta^2}{\sqrt{n \log |\univ|} \log^{1/4} k} \\
  \Longrightarrow s \gtrsim \frac{\sqrt{n \log |\univ|} \log^{1/4} k}{\eps
  \eta}.
\end{align*}
Once we have set $s$ (as a function of the other parameters) to achieve analyst
privacy, we can work on establishing data privacy.  As before, the number of bad
rounds will be
\[
R \approx sT \approx \frac{n^{3/2} \log^{3/2} |\univ| \log^{1/4} k}{\eps \eta^3}.
\]
Given this bound on the number of bad rounds, we need to set
\[
\sigma \approx \frac{\Delta \sqrt{R}}{\eps} \approx \frac{\Delta n^{3/4}
\log^{3/4} |\univ| \log^{1/8} k}{\eps \eta^{3/2}}
\]
to obtain data privacy, and
\[
\tau \approx \sigma \log k \approx \frac{\Delta n^{3/4} \log^{3/4} |\univ|
\log^{9/8} k}{\eps \eta^{3/2}}
\]
to ensure that all the update queries truly have large error on the current
hypothesis $H_{t}$.

The final error bound will come from observing that $\eta$ and $\tau$ are both
lower bounds on the error.  The error is bounded below by $\tau$ because that is
the noise threshold set by the algorithm, and $\tau$ must be larger than $\eta$
or else we cannot argue that multiplicative weights makes progress during update
rounds.  Thus setting $\eta = \tau$ will approximately minimize the error.
\fi

The final error bound we obtain (ignoring the parameters $\beta$ and $\delta$)
is
\[
  O\left( \frac{\Delta^{2/5} n^{3/10} \log^{3/10} |\univ| \log^{9/20}
k}{\eps^{2/5}} \right),
\]
which gives a non-trivial error guarantee when $\Delta \ll 1/n^{3/4}$.

\section{Conclusions}
We have shown that it is possible to privately answer many queries while also
preserving the privacy of the data analysts even if multiple analysts may
collude, or if a single analyst may register multiple accounts with the data
administrator. In the one-query-to-many-analyst privacy for linear queries in
the non-interactive setting, we are able to recover the nearly optimal
$\tilde{O}(1/\sqrt{n})$ error bound achievable without promising analyst
privacy. However, it remains unclear whether this bound is achievable for
one-analyst-to-many-analyst privacy, or for non-linear queries, or in the
interactive query release setting.

We have also introduced a novel view of the private query release problem as an
equilibrium computation problem in a two-player zero-sum game. This allows us to
encode different privacy guarantees by picking strategies of the different
players and the neighboring relationship on game matrices (i.e., differing in a
single row for analyst privacy, or differing by $1/n$ in $\ell_\infty$ norm for
data privacy).  We expect that this will be a useful point of view for other
problems. In this direction, it is known how to privately compute equilibria in
certain types of multi-player games \cite{KPRU12}. Is there a useful way to use
this multi-player generalization when solving problems in private data release,
and what does it mean for privacy?

\SUBSECTION{Acknowledgments}
We thank the anonymous reviewers, along with Jake Abernethy and Salil Vadhan,
for their helpful suggestions.

\iffull
\bibliographystyle{alpha}
\fi
\bibliography{./refs}

\iffull
\appendix

\section{Proof of Lemma~\ref{lem:sdtodp}}~\label{sec:app1}
First we restate the lemma.  Consider the following process:
\begin{itemize}
  \item Fix an $(\eps, \delta)$-differentially private mechanism $\cA \from
    \cU^* \to \cR$ and a bit $b \in \bits$.  Let $D_0 = \emptyset$.
  \item For $t = 1, \dots T$
    \begin{itemize}
      \item The (possibly randomized) adversary $\cB(y_1,\dots,y_{t}; r)$
        chooses two distributions $B^0_t, B^1_t$ such that $SD(B^0_t, B^1_t)
        \leq \sigma$.
      \item Choose $x_{t} \getsr B^b_{t}$ and let $D_t = D_{t-1} \cup
        \set{x_{t}}$.
      \item Choose $y_{t} \getsr \cA(D_{t})$.
    \end{itemize}
\end{itemize}
For a fixed mechanism $\cA$ and adversary $\cB$, let $V^0$ be the distribution
on $(y_1, \dots, y_{T})$ when $b = 0$ and $V^1$ be the distribution on $(y_1,
\dots, y_{T})$ when $b = 1$.
\begin{lemma}
  If $\eps \leq 1/2$ and $T \sigma \leq 1/12$, then with probability at least $1
  - T\delta - \delta'$ over $y = (y_1, \dots, y_T) \getsr V^0$,
  \[
    \left| \ln \left( \frac{V^0(y)}{V^1(y)} \right) \right| \leq \eps (T\sigma)
    \sqrt{2T \log(1/\delta')} + 30\eps^2(T\sigma) T.
  \]
\end{lemma}
\begin{proof}
  Given distributions $B^0, B^1$ such that $SD(B^0, B^1) \leq \sigma$, there
  exist distributions $C^0, C^1, C$ such that $B^0 = \sigma C^0 + (1-\sigma) C$
  and $B^1 = \sigma C^1 + (1-\sigma) C$.  An alternative way to sample from the
  distribution $B^b$ is to flip a coin $c \in \bits$ with bias $\sigma$, and if
  the coin comes up $1$, sample from $C^b$, otherwise sample from $C$.

  Consider a partial transcript $(r, y_1, \dots, y_{t-1})$.  Fixing the
  randomness of the adversary will fix the coins $c_1, \dots, c_{T}$, which
  determine whether or not the adversary samples from $C^b_{j}$ or $C_{j}$ for
  $j\in [T]$.  Let $w = \sum_{j=1}^{T} c_j$.  Fixing the randomness of the
  adversary and $y_{1}, \dots, y_{t-1}$ will also fix the distributions $C_{j}$
  for $j \leq t$ and, in rounds for which $c_j = 0$, will fix the samples $x_j$
  for $j \leq t$.  If we let $D^0_t, D^1_t$ denote the database $D_t$ in the
  case where $b = 0, 1$, respectively, then we have $$|D^0_t - D^1_t| \leq
  \sum_{j=1}^{t} c_{j} \leq \sum_{j=1}^{T} c_{j} = w.$$  Thus,
  \[
    \left|\ln\left(\frac{V^0_t(y_{t} | r, y_1,\dots, y_{t-1})}{V^1_t(y_{t} | r,
    y_1, \dots, y_{t-1})}\right)\right| \leq w \eps,
  \]
  and
  \[
    \ex{\ln\left(\frac{V^0_t(y_{t} | r, y_1,\dots, y_{t-1})}{V^1_t(y_{t} | r, y_1,
    \dots, y_{t-1})}\right)} \leq w\eps\min\left\{e^{w \eps} - 1, 1\right\},
  \]
  where the expectation is taken over $V^0_{t} | r, y_1, \dots, y_{t-1}$.

  Fix $w \in \set{0,\dots,T}$.  Conditioning on any $r$ such that
  $\sum_{t=1}^{T} c_t = w$, we can apply Azuma's inequality as in~\cite{DRV10}
  to obtain
  \[
    D^{T\delta + \delta'}_{\infty}(V^0 | w || V^1 | w) \leq w\eps
    \sqrt{2T\log(1/\delta')} + w \eps \min\left\{ e^{w \eps} - 1, 1\right\} T.
  \]
  Thus,
  \begin{align}
    D^{T\delta + \delta'}_{\infty}(V^0 || V^1)
    &\leq{} \sum_{w = 1}^{T} \prob{w} \left( w\eps \sqrt{2T\log(1/\delta')} + w
    \eps \min\left\{ e^{w \eps} - 1, 1\right\} T \right) \notag \\
    &={} \sum_{w = 1}^{T} \prob{w} w\eps \sqrt{2T\log(1/\delta')} +
    \sum_{w=1}^{T} \prob{w} w \eps \min\left\{ e^{w \eps} - 1, 1\right\} T.
    \label{eq:bigsum}
  \end{align}
  First, we consider the left sum in~\eqref{eq:bigsum}.
  \begin{align*}
    &\sum_{w=1}^{T} \prob{w} w\eps \sqrt{2T\log(1/\delta')} \\
    ={} &\eps \sqrt{2T\log(1/\delta')} \sum_{w = 1}^{T} \binom{T}{w} \sigma^w
    (1-\sigma)^{T-w} w \\
    ={} &\eps \sqrt{2T\log(1/\delta')} (T\sigma) \sum_{w =
    0}^{T-1}\binom{T-1}{w} \sigma^{w} (1-\sigma)^{T-1-w} \tag{$\binom{T}{w}w =
    \binom{T-1}{w-1}T$} \\
    ={} &\eps \sqrt{2T\log(1/\delta')} (T\sigma)
  \end{align*}
  Now, we work on the right sum in~\eqref{eq:bigsum}.
  \begin{align*}
    &\sum_{w = 1}^{T} \prob{w} \left( w \eps \min\left\{ e^{w \eps} - 1, 1\right\} T \right) \\
    ={} &\sum_{w = 1}^{T} \binom{T}{w} \sigma^w (1-\sigma)^{T-w} \left( w \eps
    \min\left\{ e^{w \eps} - 1, 1\right\} T \right) \\
    ={} &\left(4\eps^2T\right) \sum_{w = 1}^{1/\eps} \binom{T}{w}  \sigma^w
    (1-\sigma)^{T-w} w + (\eps T) \sum_{w = 1/\eps}^{T} \binom{T}{w}  \sigma^w
    (1-\sigma)^{T-w} w \\
    ={} &\left(4\eps^2T\right) \sum_{w = 1}^{1/\eps}
    \left(\frac{eT\sigma}{w}\right)^w w^2 + (\eps T) \sum_{w = 1/\eps}^{T}
    \left(\frac{eT\sigma}{w}\right)^w  w \\
    \leq{} &\left(4\eps^2T\right) \sum_{w = 1}^{1/\eps} \left(eT\sigma \right)^w
    + (\eps T) \sum_{w = 1/\eps}^{T} \left(eT \sigma \right)^w  \tag{$w^2/w^w
    \leq 1$ for $w \in \N$}\\
    \leq{} &4\eps^2T (2eT\sigma) + 2(eT\sigma)^{-1/\eps} \eps T \leq 3\eps^2T
    \tag{$eT\sigma \leq 1/4$} \\
    \leq{} &24 \eps^2 (T\sigma) T + 4 \eps^2 (T\sigma) T \leq 30 \eps^2 (T\sigma) T
  \end{align*}
  Combining our bounds for the left and right sums in~\eqref{eq:bigsum}
  completes the proof.
\end{proof}
\fi
\end{document}